\documentclass[journal,twoside]{IEEEtran}

\usepackage[latin1]{inputenc}
\usepackage[T1]{fontenc}
\usepackage{graphicx}
\usepackage{amssymb}
\usepackage{latexsym}
\usepackage{amsfonts}
\usepackage{amsmath}
\usepackage{color}
\usepackage{cite}
\usepackage[noend]{algorithmic}
\usepackage{algorithm}

\usepackage{mathtools}

\newcommand{\tr}{\textrm{tr}}
\newcommand{\diag}{\operatorname{diag}}

\newtheorem{thm}{Theorem}

\newtheorem{lem}{Lemma}

\newcommand{\Ptilde}{\widetilde P}
\newcommand{\vecop}{\operatorname{vec}}
\newcommand{\vect}[1]{\mathbf{#1}}
\newcommand{\CN}{\mathcal{CN}}
\newcommand{\MMSE}[1]{#1_{\mathrm{MMSE}}}
\newcommand{\MVU}[1]{#1_{\mathrm{MVU}}}

\newcommand{\Htilde}{\widetilde{\vect{H}}}
\newcommand{\Hhat}{\widehat{\vect{H}}}
\newcommand{\ejo}{e^{j \omega}}

\newcommand{\E}{\mathbb{E}}
\newcommand{\rank}{\mathrm{rank}}

\newcommand{\minimize}[1]{{\underset{{#1}}{\mathrm{minimize}}}}
\newcommand{\maximize}[1]{{\underset{{#1}}{\mathrm{maximize}}}}

\newcommand{\Ical}{\boldsymbol{\mathcal{I}}}

\begin{document}

\title{Training Sequence Design for MIMO Channels:\\ An Application-Oriented Approach}
\author{Dimitrios Katselis,$^{\dag}$\thanks{$^{\dag}$ Equally contributing authors. The order of their names is alphabetical.}$^{\star}$\thanks{$^{\star}$ Corresponding author. Email: {dimitrik@kth.se}.} Cristian R. Rojas,$^{\dag}$ Mats Bengtsson, Emil Björnson, Xavier
  Bombois, Nafiseh Shariati, Magnus Jansson and Håkan Hjalmarsson
\thanks{D.\ Katselis, C.\ R.\ Rojas, M.\ Bengtsson, E.\ Björnson, N.\ Shariati, M.\ Jansson and H.\ Hjalmarsson are with ACCESS Linnaeus Center, School of
Electrical Engineering, KTH Royal Institute of Technology, SE
100-44, Stockholm, Sweden. E-mail: dimitrik@kth.se, cristian.rojas@ee.kth.se, mats.bengtsson@ee.kth.se, emil.bjornson@ee.kth.se, nafiseh@kth.se, magnus.jansson@ee.kth.se, hjalmars@kth.se.}
\thanks{X.\ Bombois is with Delft Center for Systems and Control, Delft University of Technology, Mekelweg 2, 2628 CD, Delft, The Netherlands. E-mail: X.J.A.Bombois@tudelft.nl.}
\thanks{This work was partially supported by the Swedish Research Council under contract 621-2010-5819.}
}


\maketitle
\begin{abstract}
In this paper, the problem of training optimization for estimating
a multiple-input multiple-output (MIMO) flat fading channel in the
presence of spatially and temporally correlated Gaussian noise is
studied in an application-oriented setup. So far,
the problem of MIMO channel estimation has mostly been treated within the
context of minimizing the mean square error (MSE) 
of the channel estimate subject to various constraints,
such as an upper bound on the available training energy. 
We introduce a more general framework for the task of training
sequence design in MIMO systems, which can treat not only the
minimization of channel estimator's MSE, but also the optimization of a
final performance metric of interest related to the use of the
channel estimate in the communication system. First, we show that the proposed
framework can be used to minimize the training energy budget subject to a
quality constraint on the MSE of the channel estimator. A deterministic
version of the ``dual'' problem is also provided. We then focus on four specific
applications, where the training sequence can be optimized with
respect to the classical channel estimation MSE, a weighted channel estimation MSE and the MSE of the equalization error due to the use of an
equalizer at the receiver or an appropriate linear precoder at the
transmitter. In this way, the intended use of the channel estimate is explicitly
accounted for. The superiority of the proposed designs over existing methods is demonstrated via numerical simulations.
\end{abstract}

\begin{keywords}
Channel equalization, L-optimality criterion, MIMO channels, system identification,
training sequence design. 
\end{keywords}

\IEEEpeerreviewmaketitle


\section{Introduction}
\PARstart{A}{n} important factor in the performance of multiple
antenna systems is the accuracy of the channel state information
(CSI)~\cite{TarokhNSC:99}. CSI is primarily used at the receiver
side for purposes of coherent or semi-coherent detection, but it
can be also used at the transmitter side, e.g., for precoding and
adaptive modulation. 
Since in
communication systems the maximization of spectral efficiency is an
 objective of interest, the training duration and energy
should be minimized. Most current systems use training signals
that are white, both spatially and temporally, which is known to
be a good choice according to
several criteria~\cite{StoicaB:03,HassibiH:03}. However, in case that some prior knowledge of
the channel or noise statistics is available, it is possible to
tailor the training signal and to obtain a significantly improved
performance. Especially, several authors have studied scenarios
where long-term CSI in the form of a covariance matrix over the
short-term fading is available. So far, most proposed algorithms
have been designed to minimize the squared error of the channel
estimate,
e.g.,~\cite{KotechaS:04,WongP:04,Biguesh:06,Liu:07,KatselisKT:07,Bjornson:10}.
Alternative design criteria are used in \cite{WongP:04} and
\cite{Biguesh:09}, where the channel entropy is minimized given
the received training signal. In \cite{VikaloHHK:04}, the
resulting capacity in the case of a single-input single-output
(SISO) channel is considered, while \cite{AhmedTS:05} focuses on the
pairwise error probability.

Herein, a generic context is described, drawing from similar
techniques that have been recently proposed for training signal
design in system identification
\cite{Jansson:04a,Geversetal:06a,Hjalmarsson:09}. This context
aims at providing a unified theoretical framework, that can be
used to treat the MIMO training optimization problem in various
scenarios. Furthermore, it provides a different way of looking at
the aforementioned problem, that could be adjusted to a wide variety of estimation-related problems in communication systems.
First, we show how the problem of minimizing the training energy
subject to a quality constraint can be solved, while a ``dual'' deterministic (average design)
problem is considered\footnote{The word ``dual'' in this paper defers from the Lagrangian duality studied in the context of convex optimization theory. See  \cite{Rojas&Aguero&al:08aut} for more details on this type of duality.}. In the sequel, we
show that by a suitable definition of the performance measure the problem of optimizing the training for minimizing the
channel MSE can be treated as a special case. We also consider a weighted version of the channel MSE,
which relates to the well-known L-optimality criterion~\cite{Kiefer-74}. Moreover, we explicitly consider how
the channel estimate will be used and attempt to optimize the end
performance of the data transmission, which is not necessarily
equivalent to minimizing the mean square error (MSE) of the
channel estimate. Specifically, we study two uses of the
channel estimate: channel equalization at the receiver using a
minimum mean square error (MMSE) equalizer and channel inversion
(zero-forcing precoding) at the transmitter, and derive the
corresponding optimal training signals for each case. In the case of MMSE equalization, separate
approximations are provided for the high and low SNR regimes. Finally, the resulting
performance is illustrated based on numerical simulations. Compared
to related results in the control
literature, here we directly design a finite-length training
signal and consider not only deterministic channel parameters, but
also a Bayesian channel estimation framework. A related pilot
design strategy has been proposed in \cite{CiblatBG:08} for the
problem of jointly estimating the frequency offset and the channel
impulse response in single antenna transmissions.

Implementing an adaptive choice of pilot signals in a
practical system would require a feedback signalling overhead,
since both the transmitter and the receiver have to agree on the
choice of the pilots. Just as previous studies in the area, the
current paper is primarily intended to provide a theoretical
benchmark on the resulting performance of such a scheme. Directly
considering the end performance in the pilot design is a step into
making the results more relevant. The data model used in
\cite{KotechaS:04,WongP:04,Biguesh:06,Liu:07,KatselisKT:07,Bjornson:10,Biguesh:09}
is based on a questionable assumption, namely that the channel is
frequency flat, but that the noise is allowed to be frequency
selective. Such an assumption might be relevant in systems that
share spectrum with other radio interfaces using a narrower
bandwidth and possibly in situations where channel coding
introduces a temporal correlation in interfering signals. In order
to focus on the main principles of our proposed strategy and to
keep the mathematical derivations as simple as possible, the same
model has been used in the current paper.

As a final comment, the novelty of this paper is on introducing the application-oriented
framework as the appropriate context for training sequence design in communication systems.
To this end, Hermitian form-like approximations of performance metrics are addressed here because they usually are good approximations
of many performance metrics of interest, as well as, for simplicity purposes
and comprehensiveness of presentation. To illustrate the framework, we have for simplicity chosen to study performance metrics related to the MSE of the information carrying signal after equalization. Directly designing for performance metrics like bit error rate (BER) would be even more relevant but would involve more technical complications. Also, the BER is with good approximation monotonically increasing in the MSE of the input to the detector and we illustrate numerically that our design outperforms previous state-of-the-art also in terms of BER.

This paper is organized as follows: Section \ref{sec:model} introduces the basic MIMO received signal
model and specific assumptions on the structure of channel and noise covariance matrices. Section \ref{sec:OptEst} presents the
optimal channel estimators, when the channel is considered to be either a deterministic or a random matrix. Section \ref{sec:OptTrain} 
presents the application-oriented optimal training designs in a guaranteed performance context, based on confidence ellipsoids
and Markov bound relaxations. Moreover, Section \ref{sec:applic} focuses on four specific applications, namely that of MSE channel estimation,
channel estimation based on the L-optimality criterion and finally channel estimation for MMSE equalization and ZF precoding. Numerical simulations are provided
in Section \ref{sec:sims}, while Section \ref{sec:concl} concludes this paper.

\textbf{Notations}: Boldface (lower case) is used for column
vectors, $\vect{x}$, and (upper case) for matrices, $\vect{X}$.
Moreover, $\vect{X}^T$, $\vect{X}^H$, $\vect{X}^*$ and
$\vect{X}^\dagger$ denote the transpose, the conjugate transpose,
the conjugate and the Moore-Penrose pseudoinverse of $\vect{X}$,
respectively. The trace of $\vect{X}$ is denoted as
$\tr(\vect{X})$ and $\vect{A} \succeq \vect{B}$ means that
$\vect{A}-\vect{B}$ is positive semidefinite. $\vecop(\vect{X})$ is the vector produced by
stacking the columns of $\vect{X}$, and $(\vect{X})_{i,j}$ is the
$(i,j)$-th element of $\vect{X}$. $[\vect{X}]_+$ means that all
negative eigenvalues of $\vect{X}$ are replaced by zeros (i.e., $[\vect{X}]_+
\succeq \vect{0}$).
$\mathcal{CN}(\bar{\vect{x}},\vect{Q})$ stands for circularly
symmetric complex Gaussian random vectors, where $\bar{\vect{x}}$
is the mean and $\vect{Q}$ the covariance matrix. Finally, $\alpha!$ denotes the factorial
of the nonnegative integer $\alpha$ and ${\rm mod}(a,b)$ the modulo operation between the integers $a,b$.

\section{System Model}
\label{sec:model} We consider a MIMO communication system with
$n_T$ antennas at the transmitter and $n_R$ antennas at the
receiver.  The received signal at time $t$ is modelled as
$$
\vect{y}(t) = \vect{H} \vect{x}(t) + \vect{n}(t)
$$
where $\vect{x}(t)\in\mathbb{C}^{n_T}$ and $\vect{y}(t) \in
\mathbb{C}^{n_R}$ are the baseband representations of the
transmitted and received signals, respectively. The impact of
background noise and interference from adjacent communication
links is represented by the additive term $\vect{n}(t) \in
\mathbb{C}^{n_R}$. We will further assume that $\vect{x}(t)$ and
$\vect{n}(t)$ are independent (weakly) stationary signals. The channel
response is modeled by $\vect{H} \in \mathbb{C}^{n_R\times n_T}$,
which is assumed constant during the transmission of one block of
data and independent between blocks; that is, we are assuming
frequency flat block fading. Two different models of the channel
will be considered:
\begin{itemize}
\item[i)] A deterministic model. \item[ii)]A stochastic Rayleigh
fading model\footnote{For simplicity, we have assumed a zero-mean
channel, but it is straightforward to extend the results to Rician
fading channels, similarly to~\cite{Bjornson:10}.}, i.e., $\vecop(\vect{H})\in \CN
(\vect{0},\vect{R})$, where, for mathematical tractability, we
will assume that the known covariance matrix $\vect{R}$ possesses
the Kronecker model used, e.g., in \cite{Liu:07,Biguesh:09}:
\begin{align}
\vect{R}&=\vect{R}_T^T\otimes \vect{R}_R
\end{align}
where $\vect{R}_T\in\mathbb{C}^{n_T\times n_T}$ and
$\vect{R}_R\in\mathbb{C}^{n_R\times n_R}$ are the spatial
covariance matrices at the transmitter and receiver side,
respectively. This model has been experimentally verified in
\cite{Kermoal:02,Yu2BOMKB:04} and further motivated in \cite{Gazor:06,Rad:08}.
\end{itemize}

We consider training signals of arbitrary length $B$, represented
by $\vect{P} \in \mathbb{C}^{n_T \times B}$, whose columns are the
transmitted signal vectors during training. Placing the received
vectors in $\vect{Y}=\begin{bmatrix}\vect{y}(1) & \ldots &
  \vect{y}(B)\end{bmatrix} \in \mathbb{C}^{n_R
\times B}$, we have
$$
\vect{Y}=\vect{H} \vect{P}+\vect{N},
$$
where $\vect{N}=\begin{bmatrix}\vect{n}(1) & \ldots &
  \vect{n}(B)\end{bmatrix} \in \mathbb{C}^{n_R
\times B}$ is the combined noise and interference matrix.

Defining $\vect{\Ptilde} = \vect{P}^T\otimes \vect{I}$, we can
then write
\begin{align}
\label{m} \vecop(\vect{Y})=\vect{\Ptilde} \vecop(\vect{H})
+\vecop(\vect{N}).
\end{align}
As, for example, in \cite{Biguesh:09,Liu:07}, we assume that
$\vecop(\vect{N})\in \CN(\vect{0},\vect{S})$, where the covariance
matrix $\vect{S}$ also possesses a Kronecker structure
\begin{align}
\vect{S}&=\vect{S}_Q^T\otimes \vect{S}_R.
\end{align}
Here, $\vect{S}_Q\in\mathbb{C}^{B\times B}$ represents the
temporal covariance matrix\footnote{We set the subscript $Q$ to
$\vect{S}_Q$ to highlight its temporal nature and the fact that
its size is $B\times B$. The matrices with subscript $T$ in this
paper share the common characteristic that they are $n_T\times
n_T$, while those with subscript $R$ are $n_R\times n_R$.} and
$\vect{S}_R\in\mathbb{C}^{n_R\times
n_R}$ represents the received spatial covariance matrix.\\

The channel and noise statistics will be assumed known to the
receiver during estimation. Statistics can often be achieved by
long-term estimation and
tracking \cite{WernerJ:09}.\\

For the data transmission phase, we will assume that the transmit
signal $\{\vect{x}(t)\}$ is a zero-mean, weakly stationary
process, which is both temporally and spatially white, i.e., its
spectrum is $\vect{\Phi}_x(\omega) = \lambda_x \vect{I}$.

\section{Channel Matrix Estimation}
\label{sec:OptEst}

\subsection{Deterministic Channel Estimation}

The minimum variance unbiased (MVU) channel estimator for the
signal model \eqref{m}, subject to a deterministic channel
(Assumption i) in Section \ref{sec:model}), is given by~\cite{Kay:93}
\begin{equation}
\label{MVU}
\begin{split}
\vecop(\MVU{\widehat{\vect{H}}}) &= (\vect{\Ptilde}^H
\vect{S}^{-1} \vect{\Ptilde})^{-1}
 \vect{\Ptilde}^H \vect{S}^{-1}
\vecop(\vect{Y}).
\end{split}
\end{equation}
This estimate has the distribution
\begin{align}
\label{MVUHdist} \vecop(\MVU{\widehat{\vect{H}}}) \in \CN
(\vecop(\vect{H}), \Ical_{\text{F,MVU}}^{-1}),
\end{align}
where $\Ical_{\text{F,MVU}}$ is the inverse covariance matrix
\begin{equation}\label{fim1}
\Ical_{\text{F,MVU}} = \vect{\Ptilde}^H \vect{S}^{-1}
\vect{\Ptilde}.
\end{equation}
From this, it follows that the estimation error $\widetilde{\vect
H} \triangleq \MVU{\widehat{\vect{H}}}-\vect{H}$  will, with
probability $\alpha$, belong to the uncertainty set
\begin{equation}\label{DD}
{\cal D}_D= \left\{ \widetilde{\vect H} \ :\ \vecop^H( \widetilde{\vect
H})\Ical_{\text{F,MVU}}\vecop(\widetilde{\vect H}) \leq
\frac{1}{2} \chi^2_\alpha(2n_Tn_R) \right\},
\end{equation}
where $\chi^2_\alpha(n)$ is the $\alpha$ percentile of the
$\chi^2(n)$ distribution~\cite{Hjalmarsson:09}.

\subsection{Bayesian Channel Estimation}
\label{sec:bayes} For the case of a stochastic channel model
(Assumption ii) in Section \ref{sec:model}), the posterior channel
distribution  becomes (see \cite{Kay:93})
\begin{equation} \label{eq_channel_information}
\vecop(\vect{H}) | \vect{Y},\vect{P} \in
\mathcal{CN}(\vecop(\MMSE{\widehat{\vect{H}}}),\MMSE{\vect{C}}),
\end{equation}
where the first and second moments are
\begin{equation}
\label{MMSE}
\begin{split}
\vecop(\MMSE{\widehat{\vect{H}}}) &= (\vect{R}^{-1} +
\vect{\Ptilde}^H \vect{S}^{-1} \vect{\Ptilde})^{-1}
 \vect{\Ptilde}^H \vect{S}^{-1}
\vecop(\vect{Y}), \\
\MMSE{\vect{C}} & = (\vect{R}^{-1} + \vect{\Ptilde}^H
\vect{S}^{-1} \vect{\Ptilde})^{-1}.
\end{split}
\end{equation}
Thus, the estimation error $\widetilde{\vect H} \triangleq
\MMSE{\widehat{\vect{H}}}-\vect{H}$  will, with probability
$\alpha$, belong to the uncertainty set
\begin{equation}\label{DB}
{\cal D}_B= \left\{ \widetilde{\vect H} \ :\ \vecop^H( \widetilde{\vect
H})\Ical_{\text{F,MMSE}}\vecop(\widetilde{\vect H}) \leq
\frac{1}{2} \chi^2_\alpha(2n_Tn_R) \right\},
\end{equation}
where $\Ical_{\text{F,MMSE}} \triangleq \MMSE{\vect{C}}^{-1}$ is
the inverse covariance matrix in the MMSE case~\cite{Hjalmarsson:09}.

\section{Application-Oriented Optimal Training Design}
\label{sec:OptTrain}

In a communication system, an estimate of the channel, say
$\widehat{\vect{H}}$, is needed at the receiver to detect the data
symbols and may also be used at the transmitter to improve the
performance. Let
$J(\widetilde{\vect H},\vect{H})$ be a scalar measure of the
performance degradation at the receiver due to the estimation
error $\widetilde{\vect H}$ for a channel $\vect{H}$.
The objective of the training signal design is then to ensure that
the resulting channel estimation error $\widetilde{\vect H}$ is
such that
\begin{equation}
\label{perfcrit}
J(\widetilde{\vect H},\vect{H}) \leq \frac{1}{\gamma}
\end{equation}
for some parameter $\gamma>0$, which we call \emph{accuracy}. 
In our settings, (\ref{perfcrit}) can not be typically ensured, since the
channel estimation error is Gaussian distributed (see
\eqref{MVUHdist} and \eqref{eq_channel_information}) and,
therefore, can be arbitrarily large. However, for the MVU
estimator \eqref{MVU}, we know that, with probability $\alpha$, $\widetilde{\vect H}$ will
belong to the set ${\cal D}_D$ defined
in \eqref{DD}. Thus, we are led to training signal designs which
guarantee \eqref{perfcrit} for all channel estimation errors
$\widetilde{\vect H} \in{\cal D}_D$. One training design
problem that is based on this concept is to minimize the
required transmit energy budget subject to this constraint

\begin{align}
\label{DGPP}
\begin{array}{cl}
\mathrm{DGPP}: & \minimize{\vect{P} \in \mathbb{C}^{n_T \times B}}\ \  \tr ( \vect{P}\vect{P}^H ) \\ %
   &\text{s.t.}\ \                                J(\widetilde{\vect
H},\vect{H}) \leq \frac{1}{\gamma}\quad \forall\, \widetilde{\vect
H}\in {\cal D}_D.
\end{array}
\end{align}

Similarly, for the MMSE estimator in Subsection \ref{sec:bayes}, the
corresponding optimization problem is given as follows
\begin{align}
\label{SGPP}
\begin{array}{cl}
\mathrm{SGPP}: & \minimize{\vect{P} \in \mathbb{C}^{n_T \times B}}\ \  \tr ( \vect{P}\vect{P}^H) \\ %
& \text{s.t.}\ \                                  J(\widetilde{\vect
H},\vect{H}) \leq \frac{1}{\gamma}\quad \forall\, \widetilde{\vect
H}\in {\cal D}_B,
\end{array}
\end{align}
where ${\cal D}_B$ is defined in \eqref{DB}. We will call
\eqref{DGPP} and \eqref{SGPP}, the deterministic guaranteed
performance problem (DGPP) and the stochastic guaranteed
performance problem (SGPP), respectively. An alternative,
``dual'', problem is to maximize the accuracy $\gamma$ subject to
a constraint ${\cal P}>0$ on the transmit energy budget. For the
MVU estimator this can be written as
\begin{align}
\label{DMPP}
\begin{array}{cl}
\mathrm{DMPP}: & \maximize{\vect{P} \in \mathbb{C}^{n_T \times B}}\ \ \gamma \\ %
&\text{s.t.}\ \
J(\widetilde{\vect{H}},\vect{H}) \leq \frac{1}{\gamma} \quad
\forall\,
\widetilde{\vect{H}}\in {\cal D}_D,\\
&\ \  \ \ \ \ \tr ( \vect{P}\vect{P}^H )\leq \cal P.
\end{array}
\end{align}
We will call this problem the deterministic maximized performance
  problem (DMPP). The corresponding Bayesian problem will be
  denoted as
the stochastic maximized performance
  problem (SMPP). We will study the DGPP/SGPP in detail in this
  contribution, but the DMPP/SMPP can be
treated in similar ways. In fact, Theorem 3 in
  \cite{Rojas&Aguero&al:08aut} suggests that the solutions to
  the DMPP/SMPP are the same as for DGPP/SGPP, save for a scaling factor.

The existing work on optimal training design for MIMO channels
are, to the best of the authors knowledge, based upon standard measures on
the quality of the channel estimate, rather than on the quality of
the end-use of the channel. The framework presented in
this section can be used to treat the existing results as special
cases. Additionally, if an end performance metric is optimized,
the DGPP/SGPP and DMPP/SMPP formulations better reflect
the ultimate objective of the training design. This type of
optimal training design formulations has already been used in the
control literature, but mainly for large sample sizes
\cite{Jansson:04a,Geversetal:06a,Barenthin:06,Bombois:09a},
yielding an enhanced performance with respect to conventional
estimation-theoretic approaches. A reasonable question is to
examine if such a performance gain can be achieved in the case of
training sequence design for MIMO channel estimation, where the
sample sizes would be very small.

\emph{Remark:} Ensuring (\ref{perfcrit}) can be translated into a chance constraint of the form
\begin{equation}
\label{eq:chanceCon}
{\rm Pr}\left\{J(\widetilde{\vect H},\vect{H}) \leq \frac{1}{\gamma}\right\}\geq 1-\varepsilon
\end{equation}
for some $\varepsilon \in [0,1]$. Problems (\ref{DGPP}), (\ref{SGPP}) and (\ref{DMPP}) correspond to a
convex relaxation of this chance constraint based on confidence ellipsoids \cite{Rojas2011}, as we show in the next subsection.

\subsection{Approximating the Training Design Problems}
\label{subsec:ConfEllips-1}

A key issue regarding the above training signal design problems is
their computational tractability. In general, they are highly
non-linear and non-convex. However, for performance metrics that are sufficiently
smooth functions of the estimation error and have a minimum when the estimation error
is zero, Taylor's theorem shows that they can be well approximated by a constant plus
a quadratic term in $\widetilde{\vect
H}$. Therefore, we consider performance metrics that can
be approximated by
\begin{align}
J(\widetilde{\vect{H}},\vect{H})\approx\vecop^H(\widetilde{\vect
H})\Ical_{\text{adm}}\vecop(\widetilde{\vect
H}).\label{J(Hhat,H)}
\end{align}
For mathematical tractability, we will further assume that the Hermitian positive definite matrix $\Ical_{\text{adm}}$
can be written in Kronecker product form as $\Ical_T^T \otimes
\Ical_R$ for some matrices $\Ical_T$ and $\Ical_R$. In Section \ref{sec:applic}, we will show several examples of practically relevant performance metrics that can be approximated in this form. This means
that we can approximate the set $ \{\widetilde{\vect H} \ :\
J(\widetilde{\vect H},\vect{H})\leq 1/\gamma \}$ of all admissible
estimation errors $\widetilde{\vect H}$ by a (complex) ellipsoid
in the parameter space
\begin{align}
\label{Dadm} {\cal D}_{adm} = \{\widetilde{\vect H} \ :\
\vecop^H(\widetilde{\vect H})\gamma\Ical_{\text{adm}}
\vecop(\widetilde{\vect H}) \leq 1 \}.
\end{align}
Consequently, the DGPP \eqref{DGPP} can be approximated by
\begin{align}
\label{ADGPP}
\begin{array}{cl}
\mathrm{ADGPP}: & \minimize{\vect{P} \in \mathbb{C}^{n_T \times B}}\ \ \tr ( \vect{P}\vect{P}^H ) \\ %
&\text{s.t.}\ \                                   {\cal D}_D \subseteq
{\cal D}_{adm} .
\end{array}
\end{align}
We call this problem the approximative DGPP (ADGPP). Both
${\cal D}_D$ and ${\cal D}_{adm}$ are level sets of quadratic
functions of the channel estimation error. Rewriting \eqref{DD} so that we have the same level as in \eqref{Dadm}, we obtain
$$
{\cal D}_D= \left\{\widetilde{\vect H} \ :\ \vecop^H(\widetilde{\vect
H})\frac{2\Ical_{\text{F,MVU}}}{\chi^2_\alpha(2n_Tn_R)}
\vecop(\widetilde{\vect H}) \leq 1 \right\}.
$$
Comparing this expression with \eqref{Dadm} gives that ${\cal
  D}_D\subseteq {\cal
  D}_{adm}$ if and only if
$$
\frac{2\Ical_{\text{F,MVU}}}{\chi^2_\alpha(2n_Tn_R)} \succeq
\gamma \Ical_{\text{adm}}
$$
(for a more general result see \cite[Theorem
3.1]{Hjalmarsson:09}).

When $\Ical_{\text{adm}}$ has the form
$\Ical_{\text{adm}}=\Ical_T^T \otimes \Ical_R$, with $\Ical_T\in
\mathbb{C}^{n_T\times n_T}$ and $\Ical_R\in \mathbb{C}^{n_R\times
n_R}$, the ADGPP \eqref{ADGPP} can then be written as
\begin{align}
\label{ADGPP2}
\begin{array}{cl}
\minimize{\vect{P} \in \mathbb{C}^{n_T \times B}} & \tr ( \vect{P}\vect{P}^H ) \\ %
\text{s.t.}                                 &
\underbrace{\vect{\Ptilde}^H \vect{S}^{-1}
\vect{\Ptilde}}_{\Ical_{\text{F,MVU}}}\succeq \frac{\gamma
\chi^2_\alpha(2n_Tn_R)}{2}\, \Ical_T^T \otimes \Ical_R.
\end{array}
\end{align}
Similarly, by observing that ${\cal D}_{adm}$ only depends on the
channel estimation error, and following the derivations above, the
SGPP can be approximated by the following formulation
\begin{align}
\label{ASGPP2}
\begin{array}{cl}
\minimize{\vect{P} \in \mathbb{C}^{n_T \times B}} & \tr ( \vect{P}\vect{P}^H ) \\ %
\text{s.t.}                                 &
\underbrace{\vect{R}^{-1}+\vect{\Ptilde}^H \vect{S}^{-1}
\vect{\Ptilde}}_{\Ical_{\text{F,MMSE}}}\succeq \frac{\gamma
\chi^2_\alpha(2n_Tn_R)}{2}\, \Ical_T^T \otimes \Ical_R.
\end{array}
\end{align}
We call the last problem approximative SGPP (ASGPP).
\emph{Remarks:}
\begin{enumerate}

    \item  Several examples of the approximation (\ref{J(Hhat,H)}) are presented in Section~\ref{sec:applic}. The approximation (\ref{J(Hhat,H)}) is
    not possible for the performance metric of every application. Therefore,
    in some applications, alternative convex approximations of
    the corresponding performance metrics may have to be found.

   \item The quality of the approximation (\ref{J(Hhat,H)}) is
   characterized by its corresponding tightness to the true performance metric.
   For our purposes, when the tightness of the
   aforementioned approximation is acceptable, such an
   approximation will be desirable because it
   corresponds to a Hermitian form, therefore offering nice
   mathematical properties and tractability.

\item The sizes of ${\cal D}_D$ and ${\cal D}_{adm}$ critically depend on the parameter $\alpha$.
In practice, requiring $\alpha$ to have a value close to $1$ corresponds to adequately representing
the uncertainty set in which (approximately) all possible channel estimation errors lie.
\end{enumerate}

\subsection{The Deterministic Guaranteed Performance Problem}
\label{subsec:ConfEllips-2}

The problem formulations for ADGPP and ASGPP in \eqref{ADGPP2} and
\eqref{ASGPP2}, respectively, are similar in structure. The
solutions to these problems (and to other approximative guaranteed
performance problems) can be obtained from the following general
theorem.

\begin{thm} \label{teo:1}
Consider the optimization problem
\begin{align} \label{eq:teo1_1}
\begin{array}{cl}
\minimize{\vect{P} \in \mathbb{C}^{n \times N}} & \tr ( \vect{P}\vect{P}^H ) \\ %
\text{s.t.}                                 & \vect{P} \vect{A}^{-1} \vect{P}^H  \succeq \vect{B} %
\end{array}
\end{align}
where $\vect{A} \in \mathbb{C}^{N \times N}$ is Hermitian positive
definite, $\vect{B} \in \mathbb{C}^{n \times n}$ is Hermitian
positive semi-definite, and $N \geq \rank(\vect{B})$. An optimal
solution to \eqref{eq:teo1_1} is
{\setlength{\arraycolsep}{0.6mm}
\begin{align} \label{eq:teo1_2}
\vect{P}^{\textrm{opt}} = \vect{U}_B \vect{D}_P
\vect{U}_A^H
\end{align}}
where $\vect{D}_P \in \mathbb{C}^{n \times N}$ is a rectangular
diagonal matrix with $\sqrt{(\vect{D}_A)_{1, 1} (\vect{D}_B)_{1,
1}},\ldots, \sqrt{(\vect{D}_A)_{m, m} (\vect{D}_B)_{m, m}}$ on the
main diagonal. Here, $m=\min(n,N)$, while $\vect{U}_A$ and
$\vect{U}_B$ are unitary matrices that originate from the eigendecompositions of $\vect{A}$ and $\vect{B}$, respectively, i.e.,
\begin{equation} \label{eq:teo1_3}
\begin{split}
\vect{A} &= \vect{U}_A \vect{D}_A \vect{U}_A^H \\
\vect{B} &= \vect{U}_B \vect{D}_B \vect{U}_B^H
\end{split}
\end{equation}
and $\vect{D}_A, \vect{D}_B$ are real-valued diagonal
matrices, with their diagonal elements sorted in ascending and
descending order, respectively; that is, $0 < (\vect{D}_A)_{1,1}
\leq \ldots \leq (\vect{D}_A)_{N,N}$ and $(\vect{D}_B)_{1,1} \geq
\ldots \geq (\vect{D}_B)_{n,n} \geq 0$.

If the eigenvalues of $\vect{A}$ and $\vect{B}$ are distinct and
strictly positive, then the solution~\eqref{eq:teo1_2} is unique
up to the multiplication of the columns of $\vect{U}_A$ and
$\vect{U}_B$ by complex unit-norm scalars.
\end{thm}
\begin{proof}
The proof is given in Appendix~\ref{app:teo_1}.
\end{proof}

By the right choice of $\vect{A}$ and $\vect{B}$, Theorem
\ref{teo:1} will solve the ADGPP in \eqref{ADGPP2}. This is shown
by the next theorem (recall that we have assumed that
$\vect{S}=\vect{S}_{Q}^T \otimes \vect{S}_{R}$).
\begin{thm}
\label{teo:2} Consider the optimization problem
\begin{align} \label{eq:teo2_1}
\begin{array}{cl}
\minimize{\vect{P} \in \mathbb{C}^{n_T \times B}} & \tr ( \vect{P}\vect{P}^H ) \\ %
\text{s.t.}                                 &  \vect{\Ptilde}^H (\vect{S}_{Q}^T \otimes \vect{S}_{R})^{-1} \vect{\Ptilde}  \succeq c \Ical_T^T \otimes \Ical_R %
\end{array}
\end{align}
where $\vect{\Ptilde} = \vect{P}^T \otimes \vect{I}$,
$\vect{S}_{Q} \in \mathbb{C}^{B \times B}$, $\vect{S}_{R} \in
\mathbb{C}^{n_R \times n_R}$ are Hermitian positive definite, and
$\Ical_T  \in \mathbb{C}^{n_T \times n_T}$, $\Ical_R  \in
\mathbb{C}^{n_R \times n_R}$ are Hermitian positive semi-definite,
and $c$ is a positive constant.

If $B \geq \rank(\Ical_T)$, this problem is equivalent to
(\ref{eq:teo1_1}) in Theorem \ref{teo:1} for
$\vect{A}=\vect{S}_{Q}$ and $\vect{B}= c
\lambda_{\textrm{max}}(\vect{S}_{R} \Ical_R) \Ical_T$, where
$\lambda_{\textrm{max}}(\cdot)$ denotes the maximum eigenvalue.
\end{thm}
\begin{proof}
The proof is given in Appendix~\ref{app:cor_2}.
\end{proof}

\subsection{The Stochastic Guaranteed Performance Problem}
\label{subsec:ConfEllips-3}

Next, we will see that Theorem \ref{teo:1} can be also used to
solve the ASGPP in \eqref{ASGPP2}. In order to obtain closed-form
solutions, we need some equality relation between the Kronecker
blocks of $\vect{R}= \vect{R}_{T}^T \otimes \vect{R}_{R}$ and of
either $\vect{S}=\vect{S}_{Q}^T \otimes \vect{S}_{R}$ or
$\Ical_{\text{adm}}=\Ical_T^T \otimes \Ical_R$. For instance, it
can be $\vect{R}_{R} = \vect{S}_{R}$, which may be satisfied if
the receive antennas are spatially uncorrelated or if the signal
and interference are received from the same main direction. See
\cite{Liu:07} for details on the interpretations of these
assumptions.

The solution to ASGPP in \eqref{ASGPP2} is given by the next
theorem.
\begin{thm} \label{teo:3}
Consider the optimization problem
\begin{align} \label{eq:teo3_1}
\begin{array}{cl}
\minimize{\vect{P} \in \mathbb{C}^{n_T \times B}} & \tr ( \vect{P}\vect{P}^H ) \\ %
\text{s.t.}                                 &  \vect{R}^{-1} +\vect{\Ptilde}^H \vect{S}^{-1} \vect{\Ptilde}  \succeq c \Ical_T^T \otimes \Ical_R %
\end{array}
\end{align}
where $\vect{\Ptilde} = \vect{P}^T \otimes \vect{I}$,
$\vect{R}=\vect{R}_{T}^T \otimes \vect{R}_{R}$, and
$\vect{S}=\vect{S}_{Q}^T \otimes \vect{S}_{R}$. Here,
$\vect{R}_T \in \mathbb{C}^{n_T \times n_T}$, $\vect{R}_R  \in
\mathbb{C}^{n_R \times n_R}$, $\vect{S}_{Q} \in \mathbb{C}^{B
\times B}$, $\vect{S}_{R} \in \mathbb{C}^{n_R \times n_R}$ are
Hermitian positive definite, and $\Ical_T  \in \mathbb{C}^{n_T
\times n_T}$, $\Ical_R  \in \mathbb{C}^{n_R \times n_R}$ are
Hermitian positive semi-definite, and $c$ is a positive constant.

\begin{itemize}

\item If $\vect{R}_R=\vect{S}_R$ and $B \geq \rank([c
\lambda_{\textrm{max}}(\vect{S}_{R} \Ical_R)
\Ical_T-\vect{R}_T^{-1}]_+)$, then the problem is equivalent to
(\ref{eq:teo1_1}) in Theorem \ref{teo:1} for
$\vect{A}=\vect{S}_{Q}$ and $\vect{B}= [c
\lambda_{\textrm{max}}(\vect{S}_{R} \Ical_R)
\Ical_T-\vect{R}_T^{-1}]_+$.

\item If $\vect{R}_R^{-1}=\Ical_R$ and $B \geq \rank([ c
\Ical_T-\vect{R}_T^{-1}]_+)$, then the problem is equivalent to
(\ref{eq:teo1_1}) in Theorem \ref{teo:1} for
$\vect{A}=\vect{S}_{Q}$ and $\vect{B}=
\lambda_{\textrm{max}}(\vect{S}_{R} \Ical_R) [c
\Ical_T-\vect{R}_T^{-1}]_+$.

\item If $\vect{R}_T^{-1}=\Ical_T$ and $B \geq \rank(\Ical_T)$,
then the problem is equivalent to (\ref{eq:teo1_1}) in Theorem
\ref{teo:1} for $\vect{A}=\vect{S}_{Q}$ and
$\vect{B}=\lambda_{\textrm{max}}(\vect{S}_{R}
[c\Ical_R-\vect{R}_R]_+) \Ical_T$.
\end{itemize}
\end{thm}
\begin{proof}
The proof is given in Appendix~\ref{app:cor_2}.
\end{proof}

The mathematical difference between ADGPP and ASGPP is the
$\vect{R}^{-1}$ term that appears in the constraint of the latter.
This term has a clear impact on the structure of the optimal ASGPP
training matrix.

It is also worth noting that the solution for $\vect{R}_R=\vect{S}_R$ requires $B \geq
\rank([c \lambda_{\textrm{max}}(\vect{S}_{R} \Ical_R)
\Ical_T-\vect{R}_T^{-1}]_+)$ which means that solutions can be
achieved also for $B<n_T$ (i.e., when only the $B<n_T$ strongest
eigendirections of the channel are excited by training). In
certain cases, e.g., when the interference is temporally white
($\vect{S}_Q = \vect{I}$), it is optimal to have $B = \rank([c
\lambda_{\textrm{max}}(\vect{S}_{R} \Ical_R)
\Ical_T-\vect{R}_T^{-1}]_+)$ as larger $B$ will not decrease the
training energy usage, cf.~\cite{Bjornson:10}.

\subsection{Optimizing the Average Performance}
\label{subsec:DetPerDes}

Except from the previously presented training designs, the application-oriented design can
be alternatively given in the following deterministic ``dual'' context. If $\vect{H}$ is considered to be deterministic, then we can setup the following
optimization problem
\begin{align}
\label{DetDes}
\begin{array}{cl}
\minimize{\vect{P} \in \mathbb{C}^{n_T \times B}} & \E_{\widetilde{\vect{H}}} \left\{J(\widetilde{\vect{H}},\vect{H})\right\} \\ %
\text{s.t.}                                 & {\mathrm{tr}(\vect{P}\vect{P}^{H})}\leq \mathcal{P}.
\end{array}
\end{align}
Clearly, for the MVU estimator
\[
\E_{\widetilde{\vect{H}}} \left\{J(\widetilde{\vect{H}},\vect{H})\right\}=\mathrm{tr}\left\{\Ical_{\text{adm}}(\vect{\Ptilde}^H
\vect{S}^{-1} \vect{\Ptilde})^{-1}\right\},
\]
so problem (\ref{DetDes}) is solved by the following theorem.
\begin{thm} \label{teo:4}
Consider the optimization problem
\begin{align} \label{eq:teo4_1}
\begin{array}{cl}
\minimize{\vect{P} \in \mathbb{C}^{n_T \times B}} & \mathrm{tr}\left\{\Ical_{\text{adm}}(\vect{\Ptilde}^H
\vect{S}^{-1} \vect{\Ptilde})^{-1}\right\} \\ %
\text{s.t.}                                 &{\mathrm{tr}(\vect{P}\vect{P}^{H})}\leq \mathcal{P}  %
\end{array}
\end{align}
where $\Ical_{\text{adm}}=\Ical_T^T \otimes
\Ical_R$ as before. Set $\Ical_T'=\Ical_T^T=\vect{U}_{T}\vect{D}_T\vect{U}_T^{H}$ and
$\vect{S}_Q'=\vect{S}_Q^T=\vect{U}_{Q}\vect{D}_Q\vect{U}_Q^{H}$. Here, $\vect{U}_{T} \in \mathbb{C}^{n_T \times n_T}$, $\vect{U}_{Q} \in \mathbb{C}^{B \times B}$
are unitary matrices and $\vect{D}_T,\vect{D}_Q$ are diagonal $n_T\times n_T$ and $B\times B$ matrices containing the eigenvalues of $\Ical_T'$ and
$\vect{S}_Q'$ in descending and ascending order, respectively. Then, the optimal training matrix $\vect{P}$ equals $\left(\vect{U}_{T}\vect{D}_P\vect{U}_{Q}^H\right)^{*}$, where $\vect{D}_P$ is an $n_T\times B$ diagonal matrix with main diagonal entries equal to
$(\vect{D}_P)_{i,i}=\sqrt{\mathcal{P}\sqrt{\alpha_i}/\sum_{j=1}^{n_T}\sqrt{\alpha_j}}, i=1,2,\ldots,n_T$ ($B\geq n_T$) and $\alpha_i=(\vect{D}_T)_{i,i}(\vect{D}_Q)_{i,i}, i=1,2,\ldots,n_T$ with the aforementioned ordering.
\end{thm}
\begin{proof}
The proof is given in Appendix~\ref{app:cor_3}.
\end{proof}

\emph{Remarks:}
\begin{enumerate}
\item In the general case of a non Kronecker-structured $\Ical_{adm}$, the solution of the different designs, (\ref{ADGPP2}), (\ref{ASGPP2}) and (\ref{eq:teo4_1})  can be obtained using numerical methods like the semidefinite relaxation
approach described in \cite{KatselisRHB:11}.
\item If $\Ical_{\text{adm}}$ depends on $\vect{H}$, then in order to implement this design, the embedded $\vect{H}$ in
$\Ical_{\text{adm}}$ may be replaced by a previous channel estimate. This implies
that this approach is possible whenever the channel variations allow for such a design. This observation also applies
to the designs in the previous subsections. See also
\cite{Gerencser&Hjalmarsson:07a,Rojas&Aguero&al:08aut}, where the
same issue is discussed for other system identification
applications.
\end{enumerate}

The corresponding performance criterion for the case of the MMSE
estimator is given by
\[
\E_{\widetilde{\vect{H}},\vect{H}} \left\{J(\widetilde{\vect{H}},\vect{H})\right\}=\mathrm{tr}\left\{\Ical_{\text{adm}}(\vect{R}^{-1}+\vect{\Ptilde}^H
\vect{S}^{-1} \vect{\Ptilde})^{-1}\right\}.
\]
In this case, we can derive closed form expressions for the optimal training under assumptions similar to those made in Theorem~\ref{teo:3}. We therefore have the following result:
\begin{thm}\label{teo:5}
Consider the optimization problem
\begin{align} \label{eq:teo5_1}
\begin{array}{cl}
\minimize{\vect{P} \in \mathbb{C}^{n_T \times B}} & \mathrm{tr}\left\{\Ical_{\text{adm}}(\vect{R}^{-1}+\vect{\Ptilde}^H
\vect{S}^{-1} \vect{\Ptilde})^{-1}\right\} \\ %
\text{s.t.}                                 &{\mathrm{tr}(\vect{P}\vect{P}^{H})}\leq \mathcal{P}  %
\end{array}
\end{align}
where $\Ical_{\text{adm}}=\Ical_T^T \otimes
\Ical_R$ as before. Set
$\vect{S}_Q'=\vect{S}_Q^T=\vect{V}_{Q}\vect{\Lambda}_Q\vect{V}_Q^{H}$. Here, we assume that $\vect{V}_{Q} \in \mathbb{C}^{B \times B}$
is a unitary matrix and $\vect{\Lambda}_Q$ a diagonal $B\times B$ matrix containing the eigenvalues of
$\vect{S}_Q'$ in arbitrary order. Assume also that $\vect{R}_T'=\vect{R}_T^T$ with eigenvalue decomposition $\vect{U}_T'\vect{\Lambda}_T'\vect{U}_T'^H$. The diagonal elements of $\vect{\Lambda}_T'$ are assumed to be arbitrarily ordered.  Then, we
have the following cases
\begin{itemize}
\item $\vect{R}_R=\vect{S}_R$: We further discriminate two cases
\begin{itemize}
\item $\Ical_T=\vect{I}$: Then the optimal training is given by a straightforward adaptation of Proposition 2 in \cite{KatselisKT:07}.
\item $\vect{R}_T^{-1}=\Ical_T$: Then, the optimal training matrix $\vect{P}$ equals $\left(\vect{U}_{T}'(\pi_{\rm opt})\vect{D}_P\vect{V}_{Q}^H(\varpi_{\rm opt})\right)^{*}$, where $\pi_{\rm opt},\varpi_{\rm opt}$ stand for the optimal orderings of the eigenvalues of $\vect{R}_T'$ and $\vect{S}_Q'$, respectively. These optimal orderings are determined by Algorithm \ref{alg:optOrd} in Appendix \ref{app:cor_4}. Additionally, define the parameter $m_{*}$ as in eq. (\ref{eq:mstar}) (see Appendix \ref{app:cor_4}). Assuming in the following that, for simplicity of notation, $(\vect{\Lambda}_T')_{i,i}$'s and $(\vect{\Lambda}_Q)_{i,i}$'s have the optimal ordering,
the optimal $(\vect{D}_{P})_{j,j}, j=1,2,\ldots, m_{*}$ are given by the expression
\begin{eqnarray}
\sqrt{\frac{\mathcal{P}+\sum_{i=1}^{m_{*}}\frac{(\vect{\Lambda}_Q)_{i,i}}{(\vect{\Lambda}_T')_{i,i}}}{\sum_{i=1}^{m_{*}}\sqrt{\frac{(\vect{\Lambda}_Q)_{i,i}}{(\vect{\Lambda}_T')_{i,i}}}}\sqrt{\frac{(\vect{\Lambda}_Q)_{j,j}}{(\vect{\Lambda}_T')_{j,j}}}-\frac{(\vect{\Lambda}_Q)_{j,j}}{(\vect{\Lambda}_T')_{j,j}}},\nonumber
\end{eqnarray}
while $(\vect{D}_{P})_{j,j}=0$ for $j=m_{*}+1,\ldots, n_T$.
\end{itemize}
\end{itemize}
\end{thm}
\begin{proof}
The proof is given in Appendix~\ref{app:cor_4}.
\end{proof}

\emph{Remarks}: Two interesting additional cases complementing the last theorem are the following:
\begin{enumerate}
\item If the modal matrices of $\vect{R}_R$ and $\vect{S}_R$ are the same,
 $\Ical_T=\vect{I}$ and $\Ical_R=\vect{I}$, then the optimal training is given by \cite{Bjornson:10}.

\item In any other case (e.g., if $\vect{R}_R\neq\vect{S_R})$, the training can be found using numerical methods like the semidefinite relaxation
approach described in \cite{KatselisRHB:11}. Note again that this approach can also handle general $\Ical_{adm}$, not necessarily expressed as $\Ical_T^T\otimes \Ical_R$.
\end{enumerate}

As a general conclusion, the objective function of the dual deterministic problems presented in this subsection
can be shown to correspond to Markov bound approximations of the chance constraint (\ref{eq:chanceCon}). According to the analysis
in \cite{Rojas2011}, these approximations should be tighter than the approximations based on confidence ellipsoids presented in Subsections \ref{subsec:ConfEllips-1},
\ref{subsec:ConfEllips-2} and \ref{subsec:ConfEllips-3}, for
practically relevant
values of $\varepsilon$.

\section{Applications}
\label{sec:applic}


\subsection{Optimal Training for Channel Estimation}
\label{subsec:ChEst}

We now consider the channel estimation problem in its standard
context, where the performance metric of interest is the (mean) square error of the
corresponding channel estimator. Linear estimators for
this task are given by (\ref{MVU}), (\ref{MMSE}). The performance
metric of interest is
\[
J(\widetilde{\vect{H}},\vect{H})=\vecop^H(\widetilde{\vect H})
\vecop(\widetilde{\vect H}),
\]
which corresponds to $\Ical_{\text{adm}}=\vect{I}$, i.e., to
$\Ical_T=\vect{I}$ and $\Ical_R=\vect{I}$. The ADGPP and ASGPP are
given by (\ref{ADGPP2}) and (\ref{ASGPP2}), respectively, with the
corresponding substitutions. Their solutions follow directly from
Theorems \ref{teo:2} and \ref{teo:3}, respectively. To the best of
the authors' knowledge, such formulations for the classical MIMO training design problem
are presented here for the first time. Furthermore, solutions to the standard approach
of minimizing the channel MSE subject to a constraint on the training energy budget are
provided by Theorems \ref{teo:4} and \ref{teo:5} as special cases.

\emph{Remark:}
      Although the confidence ellipsoid and Markov bound approximations are generally different~\cite{Rojas2011}, in the simulation
    section we show that their performance is almost identical for reasonable operating $\gamma$-regimes in the specific case of standard channel estimation.

\subsection{Optimal Training for the L-Optimality Criterion}
\label{subsec:L-optimalityChEst}

Consider now a performance metric of the form
\[
J_W(\widetilde{\vect{H}},\vect{H})=\vecop^H(\widetilde{\vect H})\vect{W}
\vecop(\widetilde{\vect H}),
\]
for some positive semidefinite weighting matrix $\vect{W}$. Assume also that
$\vect{W}=\vect{W}_1\otimes \vect{W}_2$ for some positive semidefinite matrices $\vect{W}_1,\vect{W}_2$. Taking the expected value
of this performance metric with respect to either $\widetilde{\vect{H}}$ or both $\widetilde{\vect{H}}$ and $\vect{H}$ leads to
the well-known L-optimality criterion for optimal experiment design in statistics~\cite{Kiefer-74}. In this case, $\Ical_T=\vect{W}_1^T$ and
$\Ical_R=\vect{W}_2$. In the context of MIMO communication systems, such a performance metric may arise, e.g., if we want to estimate the MIMO channel  having some deficiencies in either the transmit and/or the receive antenna arrays. The simplest case would be both $\vect{W}_1$ and $\vect{W}_2$ being diagonal with nonzero entries in the interval $[0,1]$, $\vect{W}_1$ representing the deficiencies in the transmit antenna array and $\vect{W}_2$ in the receive array. More general matrices can be considered if we assume cross-couplings between the transmit and/or receive antenna elements.

\emph{Remark}: The numerical approach of~\cite{KatselisRHB:11} mentioned after Theorems \ref{teo:4} and \ref{teo:5} can handle general weighting matrices $\vect{W}$, not necessarily Kronecker-structured.

\subsection{Optimal Training for Channel Equalization}
\label{subsec:ChEq}

In this subsection we consider the problem of
estimating a transmitted signal sequence $\{\vect{x}(t)\}$ from
the corresponding received signal sequence $\{\vect{y}(t)\}$.
Among a wide range of methods that are available
\cite{PaulrajNG:03,Verdu:98}, we will consider the MMSE equalizer
and for mathematical tractability we will approximate it by the
non-causal Wiener filter. Note that for reasonably long block
lengths, the MMSE estimate becomes similar to the non-causal
Wiener filter~\cite{Haykin:01}. Thus, the optimal training design
based on the non-causal Wiener filter should also provide good
performance when using an MMSE equalizer.

\subsubsection{Equalization using exact channel state information}

Let us first assume that $\vect{H}$ is available. In this ideal
case, and with the transmitted signal being weakly stationary with
spectrum $\vect{\Phi}_x$, the MSE-optimal estimate of the transmitted
signal $\vect{x}(t)$ from the received observations of
$\vect{y}(t)$ can be obtained according to
\begin{equation}
    \hat{\vect{x}}(t; \vect{H}) = \vect{F}(q;\vect{H}) \vect{y}(t)
\end{equation}
where $q$ is the unit time shift operator, $[q
\vect{x}(t)=\vect{x}(t+1)]$, and the non-causal Wiener filter
$\vect{F}(e^{j\omega};\vect{H})$ is given by
\begin{equation}\label{ncWiener}
\begin{split}
    \vect{F}(e^{j\omega};\vect{H}) &= \vect{\Phi}_{xy}(\omega)\vect{\Phi}_{y}^{-1}(\omega) \\ &= \vect{\Phi}_x(\omega)\vect{H}^H\left(\vect{H}\vect{\Phi}_x(\omega)\vect{H}^H +
    \vect{\Phi}_n(\omega)\right)^{-1}.
\end{split}
\end{equation}
Here, $\vect{\Phi}_{xy}(\omega)=\vect{\Phi}_x(\omega)\vect{H}^H$
denotes the cross-spectrum between $\vect{x}(t)$ and
$\vect{y}(t)$, and
\begin{equation}\label{phiy}
 \vect{\Phi}_y(\omega) = \vect{H}\vect{\Phi}_x(\omega)\vect{H}^H+\vect{\Phi}_n(\omega)
 \end{equation}
 is the spectral density of $\vect{y}(t)$. Using our assumption that
$\vect{\Phi}_x(\omega) = \lambda_x \vect{I}$, we obtain the
simplified expression
\begin{eqnarray}
    \vect{F}(e^{j\omega};\vect{H}) &=&  \vect{H}^H\left(\vect{H}\vect{H}^H +
    \vect{\Phi}_n(\omega)/\lambda_x\right)^{-1}.\label{eq:MMSEeq1}
\end{eqnarray}

\emph{Remark:}
    Assuming nonsingularity of $\vect{\Phi}_n(\omega)$ for every $\omega$, the
    MMSE equalizer is applicable for all values of the pair
    $(n_T,n_R)$.

\subsubsection{Equalization using a channel estimate}

Consider now the situation where the exact channel $\vect{H}$
is unavailable, but we only have an estimate $\widehat{\vect{H}}$. When
we replace $\vect{H}$ by its estimate in the expressions above, the
estimation error for the equalizer will increase. While the
increase in the bit error rate would be a natural measure of the
quality of the channel estimate $\widehat{\vect{H}}$,
for simplicity we consider the total MSE of the
difference, $\hat{\vect{x}}(t;\vect{H}+\widetilde{\vect H}) -
\hat{\vect{x}}(t; \vect{H}) = \vect{\Delta}(q;\widetilde{\vect
H},\vect{H}) \vect{y}(t)$ (note that
$\widehat{\vect{H}}=\vect{H}+\widetilde{\vect H}$), using the
notation $\vect{\Delta}(q;\widetilde{\vect H},\vect{H}) \triangleq
\vect{F}(q; \vect{H}+\widetilde{\vect H})-\vect{F}(q;\vect{H})$.
In view of this, we will use the channel equalization (CE)
performance metric
\begin{align}\label{Jdef}
 &   J_{CE}(\widetilde{\vect H},\vect{H})=\E \left\{ [\vect{\Delta}(q;\widetilde{\vect H},\vect{H})\vect{y}(t)]^H[\vect{\Delta}(q;\widetilde{\vect H},\vect{H})\vect{y}(t)]  \right\}
\notag \\
& = \E\left\{\tr\left( [\vect{\Delta}(q;\widetilde{\vect
H},\vect{H})\vect{y}(t)][ \vect{\Delta}(q;\widetilde{\vect
H},\vect{H})\vect{y}(t)]^H\right) \right\}
\notag \\
& =  \frac{1}{2\pi} \int_{-\pi}^{\pi} \tr\left(
\vect{\Delta}(e^{j\omega};\widetilde{\vect
H},\vect{H})\vect{\Phi}_y(\omega)\vect{\Delta}^H(e^{j\omega};\widetilde{\vect
H},\vect{H})\right)\; d\omega.
\end{align}
We see
that the poorer the accuracy of the estimate, the larger the
performance metric $J_{CE}(\widetilde{\vect H},\vect{H})$ and,
thus, the larger the  performance loss of the equalizer.
Therefore, this performance metric is a reasonable candidate to use
when formulating our training sequence design problem. Indeed, the
Wiener equalizer based on the estimate
$\widehat{\vect{H}}=\vect{H}+\widetilde{\vect H}$ of $\vect{H}$
can be deemed to have a satisfactory performance if
$J_{CE}(\widetilde{\vect H},\vect{H})$ remains below some
user-chosen threshold. Thus, we will use $J_{CE}$ as $J$ in
problems \eqref{DGPP} and \eqref{SGPP}. 
Though these problems are
not convex,  we show in Appendix~\ref{app:quad} how they
can be convexified, provided some approximations are made.

\emph{Remarks:}
\begin{enumerate}
    \item The excess MSE $J_{CE}(\widetilde{\vect H},\vect{H})$
    quantifies the distance of the MMSE equalizer using the channel
    estimate $\widehat{\vect{H}}$ over the \emph{clairvoyant} MMSE
    equalizer, i.e., the one using the true channel. This
    performance metric is not the same as the classical MSE in the
    equalization context, where the difference $\hat{\vect{x}}(t;\vect{H}+\widetilde{\vect H}) -
\vect{x}(t)$ is considered instead of
$\hat{\vect{x}}(t;\vect{H}+\widetilde{\vect H}) -
\hat{\vect{x}}(t; \vect{H})$. However, since in practice the best
transmit vector estimate that can be attained is the clairvoyant
one, the choice of $J_{CE}(\widetilde{\vect H},\vect{H})$  is
justified. This selection allows for a performance metric approximation given by
(\ref{J(Hhat,H)}).
   \item There are certain cases of interest, where $J_{CE}(\widetilde{\vect
   H},\vect{H})$ approximately coincides with the classical equalization MSE.
   Such a case occurs when $n_R\geq n_T$, $\vect{H}$ is full column rank and the SNR is high
   during data transmission.
\end{enumerate}

\subsection{Optimal training for Zero-Forcing (ZF) Precoding}
\label{subsec:ZFpc}

Apart from receiver side channel equalization, as another example
of how to apply the channel estimate we consider point-to-point
zero-forcing precoding, also known as channel
inversion~\cite{HochwaldPS:05a}. Here the channel estimate is fed
back to the transmitter and its (pseudo-)inverse is used as a
linear precoder. The data transmission is described by
\begin{equation*}
  \vect{y}(t)=\vect{H}\vect{\Psi}\vect{x}(t) +
  \vect{v}(t)
\end{equation*}
where the precoder is $\vect{\Psi}=\Hhat^\dagger$, i.e.,
$\vect{\Psi}=\Hhat^H(\Hhat \Hhat^H)^{-1}$ if we limit ourselves to
the practically relevant case $n_T \geq n_R$ and assume that
$\Hhat$ is full rank. Note that $\vect{x}(t)$ is an $n_R\times 1$
vector in this case, but the transmitted vector is
$\vect{\Psi}\vect{x}(t)$, which is $n_T\times 1$.

Under these assumptions, and following the same strategy and
notation as in Appendix~\ref{app:quad}, we get
\begin{multline}
  \vect{y}(t;\Hhat) - \vect{y}(t; \vect{H}) = \vect{H} \Hhat^\dagger
  \vect{x}(t) + \vect{v} - (\vect{H} \vect{H}^\dagger \vect{x}(t) +
  \vect{v}) \\
  = (\Hhat\Hhat^\dagger - \Htilde\Hhat^\dagger - \vect{I}) \vect{x}(t)
  \simeq -\Htilde \vect{H}^\dagger\vect{x}(t)
\end{multline}
Consequently, a quadratic approximation of the cost function is
given by
\begin{eqnarray}
  \label{eq:costChannelInv}
  J_{\text{ZF}}(\widetilde{\vect H},\vect{H})&=& \E\left\{ [\vect{y}(t;\Hhat) -
    \vect{y}(t; \vect{H})]^H [\vect{y}(t;\Hhat) - \vect{y}(t;
    \vect{H})] \right\}\nonumber \\
  &\simeq& \lambda_x \vecop^H(\Htilde) \left((\vect{H}^\dagger
    (\vect{H}^\dagger)^H)^T \otimes \vect{I}\right) \vecop(\Htilde) \nonumber \\
  &=& \vecop^H(\Htilde) (\Ical_T^T \otimes \Ical_R) \vecop(\Htilde),
\end{eqnarray}
if we define $\Ical_T = \lambda_x \vect{H}^\dagger
(\vect{H}^\dagger)^H = \lambda_x \vect{H}^H  (\vect{H}
\vect{H}^H)^{-2} \vect{H}$ and $\Ical_R=\vect{I}$.

\emph{Remark:} The cost functions of (\ref{eq:teo4_1}) and (\ref{eq:teo5_1}) reveal the fact that any performance-oriented
training design is a compromise between the strict channel estimation accuracy and the desired
accuracy related to the end performance metric at hand. Caution is needed to identify cases where
the performance-oriented design may severely degrade the channel estimation accuracy, annihilating
all gains from such a design. In the case of ZF precoding, if $n_T>n_R$, $\Ical_T$ will have rank at most
$n_R$ yielding a training matrix $\vect{P}$ with only $n_R$ active eigendirections. This is in contrast to the secondary
target, which is the channel estimation accuracy. Therefore, we expect ADGPP, ASGPP and the approaches in Subsection~\ref{subsec:DetPerDes}
to behave abnormally in this case. Thus, we propose the performance-oriented design only when $n_T=n_R$ in the context of the ZF precoding.

\section{Numerical Examples}
\label{sec:sims}

\begin{figure}
  \centering
  \includegraphics[width=7cm]{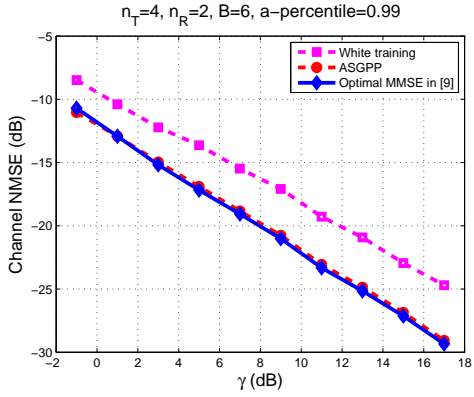}\\
  \caption{$n_T=4, n_R=2, B=6, a(\%)=99$: Channel Estimation NMSE based on Subection~\ref{subsec:ChEst} with $\vect{R}_R= \vect{S}_R$.}\label{fig:1}
\end{figure}

\begin{figure}
  \centering
  \includegraphics[width=7cm]{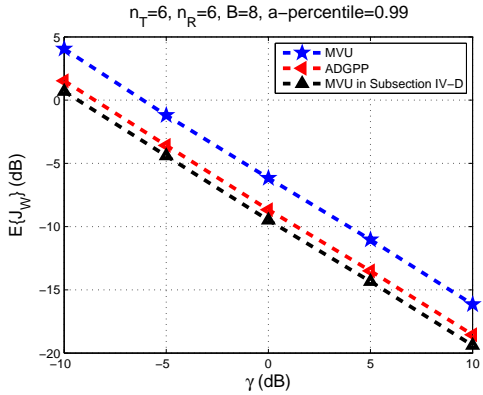}\\
\caption{$n_T=6, n_R=6, B=8, a(\%)=99$: L-optimality criterion with arbitrary but positive-semidefinite $\vect{W}_1, \vect{W}_2$ for the MVU estimator.}\label{fig:2}
\end{figure}

\begin{figure}
  \centering
  \includegraphics[width=7cm]{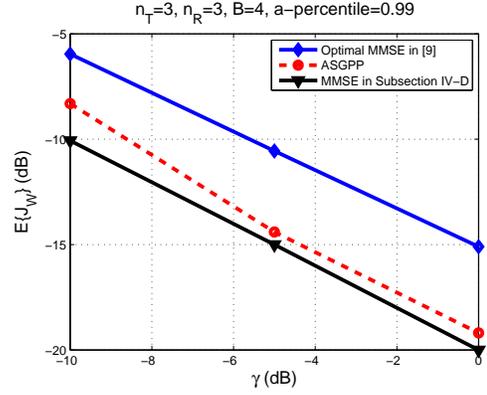}\\
  \caption{$n_T=3, n_R=3, B=4, a(\%)=99$: L-optimality criterion with arbitrary but positive-semidefinite $\vect{W}_1, \vect{W}_2$ for the MMSE estimator with $\vect{R}_R= \vect{S}_R$.}\label{fig:3}
\end{figure}

\begin{figure}
  \centering
  \includegraphics[width=7cm]{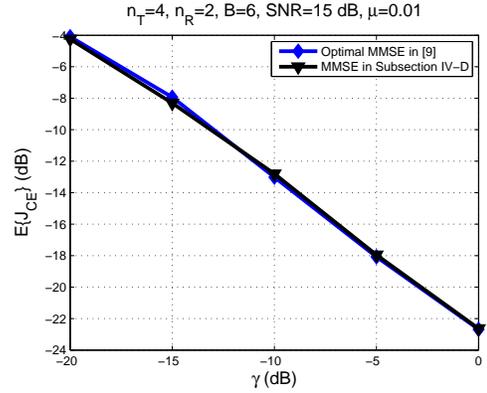}\\
  \caption{$n_T=4, n_R=2, B=6, {\rm SNR}=15 {\rm dB}, \mu=0.01$: MMSE Channel Equalization with $\vect{R}_R\neq \vect{S_R}$.}\label{fig:4}
\end{figure}

\begin{figure}
  \centering
  \includegraphics[width=7cm]{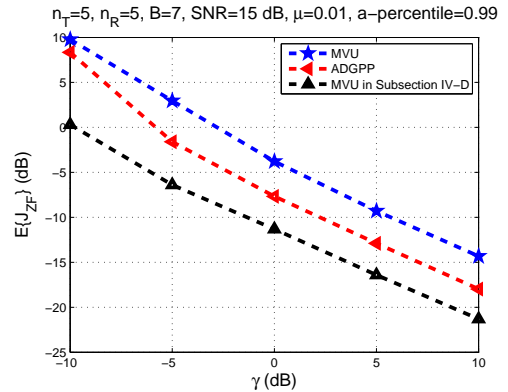}\\
  \caption{$n_T=5, n_R=5, B=7, {\rm  SNR}=15 {\rm dB}, a(\%)=99, \mu=0.01$: ZF precoding based on Subection~\ref{subsec:ZFpc} for the MVU estimator. $\Ical_{\text{adm}}$ is based on a previous
  channel estimate.}\label{fig:5}
\end{figure}

\begin{figure}
  \centering
  \includegraphics[width=7cm]{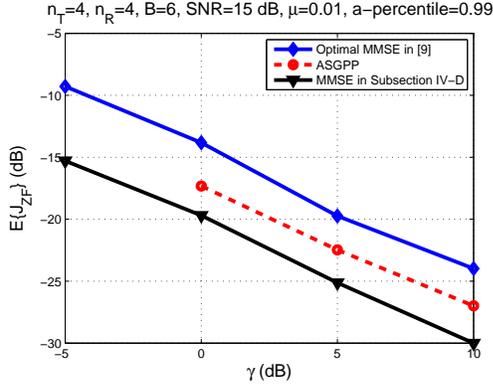}\\
  \caption{$n_T=4, n_R=4, B=6, {\rm  SNR}=15 {\rm dB}, \mu=0.01, a(\%)=99$: ZF precoding MSE based on Subection~\ref{subsec:ZFpc}  for the MMSE estimator with $\vect{R}_R= \vect{S}_R$. $\Ical_{\text{adm}}$ is based on a previous
  channel estimate.}\label{fig:6}
\end{figure}

\begin{figure}
  \centering
  \includegraphics[width=7cm]{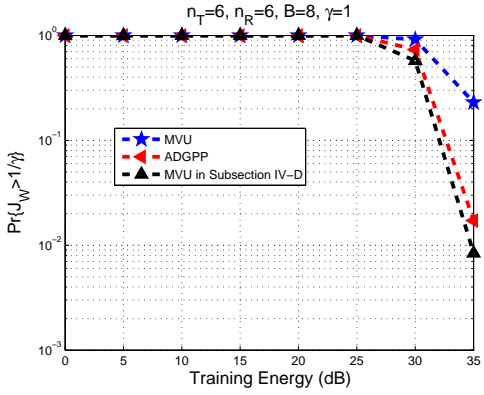}\\
  \caption{$n_T=6, n_R=6, B=8, \gamma=1$: Outage probability for the
    L-optimality criterion with the MVU estimator. The accuracy
    parameter is $\gamma=1$.}\label{fig:7}
\end{figure}

\begin{figure}
  \centering
  \includegraphics[width=7cm]{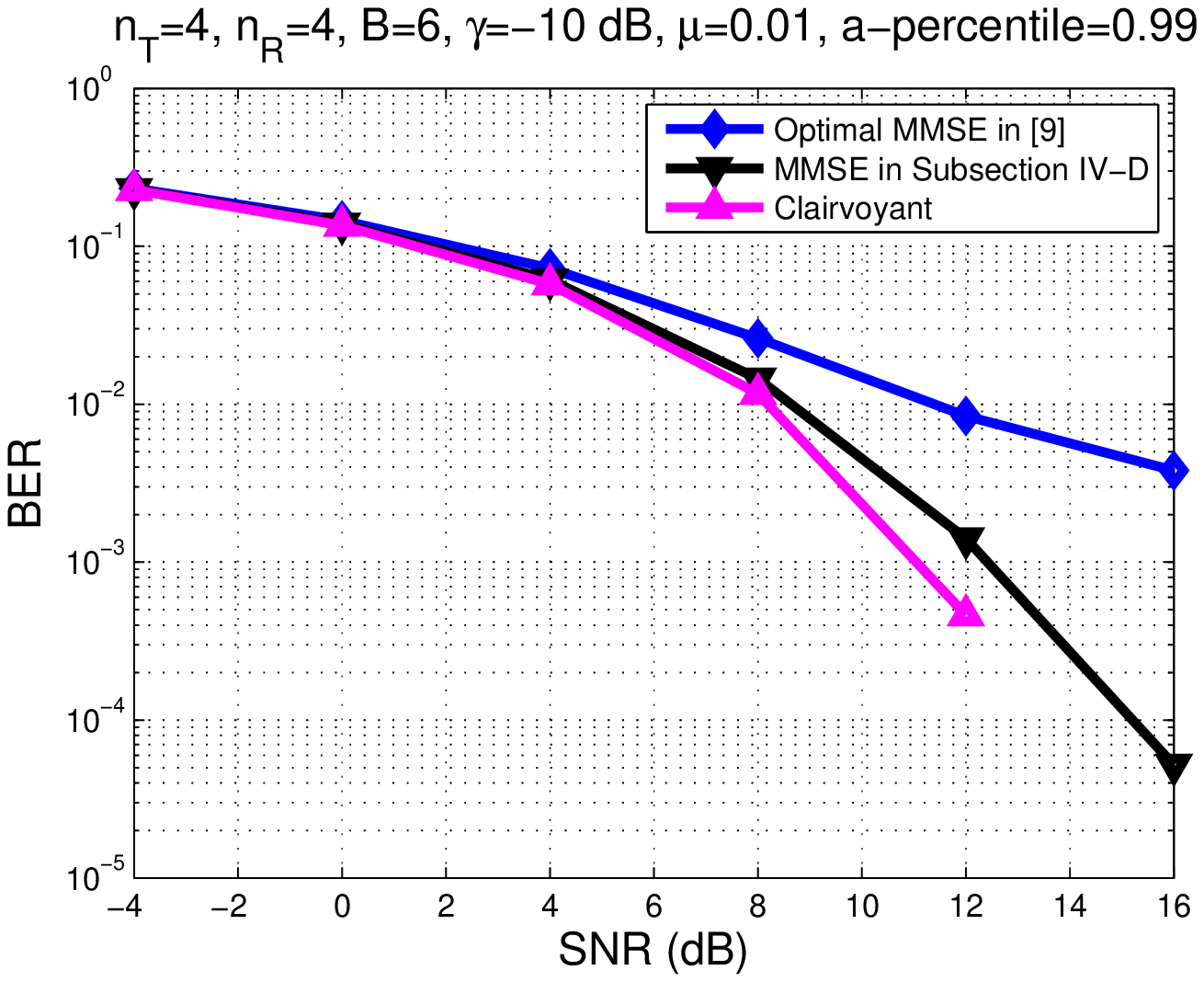}\\
  \caption{$n_T=4, n_R=4, B=6, \gamma=-10 {\rm dB}, \mu=0.01, a(\%)=99$: BER performance using the signal estimates produced by the corresponding schemes in Fig.~\ref{fig:6} with $\vect{R}_R= \vect{S}_R$ and $\gamma=-10$~dB.  $\Ical_{\text{adm}}$ is based on a previous
  channel estimate.}\label{fig:8}
\end{figure}

\begin{figure}
  \centering
  \includegraphics[width=7cm]{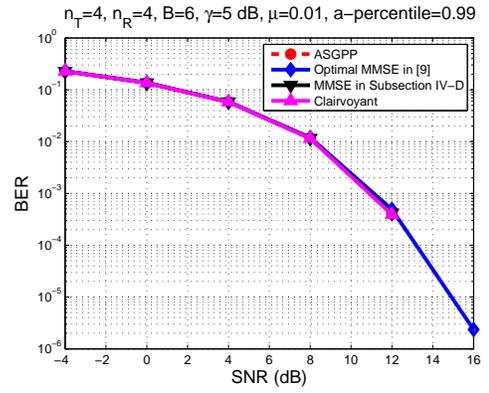}\\
  \caption{$n_T=4, n_R=4, B=6, \gamma=5 {\rm dB}, \mu=0.01, a(\%)=99$: BER performance using the signal estimates produced by the corresponding schemes in Fig.~\ref{fig:6} with $\vect{R}_R= \vect{S}_R$ and $\gamma=5$~dB.  $\Ical_{\text{adm}}$ is based on a previous
  channel estimate.}\label{fig:9}
\end{figure}

The purpose of this section is to examine the performance of
optimal training sequence designs, and compare them with existing
methods. For the channel estimation MSE figure, we plot the
normalized MSE (NMSE), i.e.,
$\E(\|\mathbf{H}-\mathbf{\widehat{H}}\|^{2}/\|\mathbf{H}\|^{2})$,
versus the accuracy parameter $\gamma$.
In all figures, fair comparison among the presented
schemes is ensured via training energy equalization. Additionally, the matrices $\vect{R}_T,
\vect{R}_R, \vect{S}_Q, \vect{S}_R$ follow the exponential model,
that is, they are built according to
\begin{equation}
(\vect{R})_{i,j}=r^{j-i},\;\; j\geq i, \label{eq:expMod}
\end{equation}
where $r$ is the (complex) normalized correlation coefficient with
magnitude $\rho=|r|<1$. We choose
to examine the high correlation scenario for all the presented
schemes. Therefore, in all plots $|r|=0.9$ for all matrices
$\vect{R}_T, \vect{R}_R, \vect{S}_Q, \vect{S}_R$. Additionally,
the transmit SNR during data transmission is chosen to be $15$~dB, when channel equalization and
ZF precoding are considered.
High SNR expressions are therefore used for optimal training
sequence designs. Since the optimal pilot sequences depend on the true
channel, we have for these two applications additionally assumed that the channel changes from block to block according to the relationship
$\vect{H}_i=\vect{H}_{i-1}+\mu \vect{E}_i$, where $\vect{E}_i$ has the same Kronecker structure as $\vect{H}$ and it is completely independent
from $\vect{H}_{i-1}$. The estimated $\vect{H}_{i-1}$ is used in the pilot design. In Figs.~\ref{fig:4}, \ref{fig:5}, \ref{fig:6}, \ref{fig:8} and \ref{fig:9} the value of $\mu$ is $0.01$.

In Fig.~\ref{fig:1} the channel estimation NMSE
performance versus the accuracy $\gamma$ is presented for three
different schemes. The scheme `ASGPP' is the optimal Wiener
filter together with the optimal guaranteed performance training
matrix described in Subsection~\ref{subsec:ChEst}. `Optimal MMSE in \cite{Bjornson:10}'
is the scheme presented in~\cite{Bjornson:10}, which solves the
optimal training problem for the \emph{vectorized} MMSE,
operating on $\mathrm{vec}(\vect{Y})$. This solution is a special case
in the statement of Theorem \ref{teo:5} for $\Ical_{adm}=\vect{I}$, i.e., $\Ical_T=\vect{I}$ and $\Ical_R=\vect{I}$. Finally, the scheme `White training'
corresponds to the use of the vectorized MMSE filter at the
receiver, with a white training matrix, i.e., one having equal
singular values and arbitrary left and right singular matrices.
This scheme is justified when the receiver knows the involved
channel and noise statistics, but does not want to sacrifice
bandwidth to feedback the optimal training matrix to the
transmitter. This scheme is also justified in fast fading
environments. In Fig.~\ref{fig:1}, we assume that $\vect{R}_R=\vect{S}_R$ and we implement the corresponding optimal training design for each scheme. 
`ASGPP' is implemented first for a certain
value of $\gamma$ and the rest of the schemes are forced to have
the same training energy. The `Optimal MMSE in \cite{Bjornson:10}' and `ASGPP'
schemes have the best and almost identical MSE performance. This indicates that for the problem
of training design with the classical channel estimation MSE, the confidence ellipsoid relaxation
of the chance constraint and the relaxation based on the Markov bound in Subsection \ref{subsec:DetPerDes} deliver
almost identical performances.

Figs.~\ref{fig:2} and \ref{fig:3} demonstrate the
L-optimality average performance metric $E\{J_{W}\}$ versus $\gamma$. Fig.~\ref{fig:2} corresponds to the L-optimality
criterion based on MVU estimators and Fig.~\ref{fig:3} is based on MMSE estimators.
In Fig.~\ref{fig:2}, the scheme `MVU' corresponds to the optimal training for channel estimation when the MVU estimator is used. This training is given by
Theorem \ref{teo:4} for $\Ical_{adm}=\vect{I}$, i.e., $\Ical_T=\vect{I}$ and $\Ical_R=\vect{I}$. `MVU in Subsection \ref{subsec:DetPerDes}' is again the MVU estimator
 based on the same theorem but for the correct $\Ical_{adm}$. The scheme `MMSE in Subsection \ref{subsec:DetPerDes}'
is given by the numerical solution mentioned below Theorem \ref{teo:5}, since $\vect{W}_1$ is different than the cases where a closed form solution is possible. Figs.~\ref{fig:2} and \ref{fig:3} clearly show that both the confidence ellipsoid and Markov bound
approximations are better than the optimal training for standard channel estimation. Therefore, for this problem the application-oriented training design is superior compared to training designs with respect to the quality of the channel estimate. 

Fig.~\ref{fig:4} demonstrates the performance of optimal training designs for the MMSE estimator in the context of MMSE channel equalization. We
assume that $\vect{R}_R\neq \vect{S}_R$, since the high SNR expressions for $\Ical_{adm}$ in the context of MMSE channel equalization in Appendix \ref{app:quad} indicate that $\Ical_T=\vect{I}$ for this application and according to Theorem \ref{teo:5} the optimal training corresponds to the optimal training for channel estimation in \cite{KatselisKT:07}. We observe that the curves almost coincide. Moreover, it can be easily verified that for MMSE channel equalization with the MVU estimator, the optimal training designs given by Theorems \ref{teo:2} and \ref{teo:4} differ slightly only in the optimal power loading. 
These
observations essentially show that the optimal training designs for the MVU and MMSE
estimators in the classical channel estimation setup are nearly optimal for the application of MMSE channel equalization. This relies on the fact that
for this particular application, $\Ical_T=\vect{I}$ in the high data transmission SNR regime.

Figs.~\ref{fig:5} and \ref{fig:6} present the corresponding performances in the case of the ZF precoding. The descriptions of the schemes are as before.
In Fig.~\ref{fig:6}, we
assume that $\vect{R}_R=\vect{S}_R$. The superiority of the application-oriented designs for the ZF precoding application is apparent in these plots. Here, $\Ical_T\neq \vect{I}$ and this is why the optimal training for the channel estimate works less well in this application. Moreover, the ``ASGPP'' is plotted for $\gamma\geq 0$~dB, since for smaller values of $\gamma$ all the eigenvalues of $\vect{B}= [c
\lambda_{\textrm{max}}(\vect{S}_{R} \Ical_R)
\Ical_T-\vect{R}_T^{-1}]_+$ are equal to zero for this particular set of parameters defining Fig.~\ref{fig:6}.

Fig.~\ref{fig:7} presents an outage plot in the context of the L-optimality criterion for the MVU estimator. We assume that $\gamma=1$. We plot
${\rm Pr}\left\{J_W>1/\gamma\right\}$ versus the training power. This plot indirectly verifies that the confidence ellipsoid relaxation of the chance constraint
given by the scheme ``ASGPP'' is not as tight as the Markov bound approximation given by the scheme ``MVU in Subsection \ref{subsec:DetPerDes}''.

Finally, Figs.~\ref{fig:8} and \ref{fig:9} present the BER performance of the nearest neighbor rule applied to the signal estimates produced by the corresponding schemes in Fig.~\ref{fig:6}, when the QPSK modulation is used. The ``Clairvoyant'' scheme corresponds to the ZF precoder with perfect channel knowledge. The channel estimates have been obtained for $\gamma=-10$ and $5$~dB, respectively. Even if the application-oriented estimates are not optimized for the BER performance metric, they lead to better performance than the `Optimal MMSE in \cite{Bjornson:10}' scheme as is apparent in Fig.~\ref{fig:8}. In Fig.~\ref{fig:9}, the performances of all schemes approximately coincide. This is due to the fact that for $\gamma=5$~dB all channel estimates are very good, thus leading to symbol MSE performance differences that have negligible impact on the BER performance.

\section{Conclusions}
\label{sec:concl}
In this contribution, we have presented a quite general framework
for MIMO training sequence design subject to flat and block
fading, as well as spatially and temporally correlated Gaussian
noise. The main contribution has been to incorporate the objective
of the channel estimation into the design.
We have shown that by a suitable approximation of
$J(\widetilde{\vect{H}},\vect{H})$, it is possible to solve this
type of problem for several interesting applications such
as standard MIMO channel estimation, L-optimality criterion, MMSE channel equalization and ZF precoding. For these problems, we have numerically demonstrated
the superiority of the schemes derived in this paper. Additionally, the proposed framework is valuable since it provides a
universal way of posing different estimation-related problems in
communication systems. We have seen that it shows interesting promise for, e.g., ZF precoding and it may yield even greater end
performance gains in estimation problems related to communication
systems, when approximations can be avoided, depending on the end
performance metric at hand.

\appendices

\section{Approximating the performance measure for MMSE Channel Equalization} \label{app:quad}

In order to obtain the approximating set ${\cal D}_{adm}$, let us
first denote the integrand in the performance metric \eqref{Jdef}
by 
 \begin{equation}\label{jprim}
  J'(\omega;  \widetilde{\vect H},\vect{H})= \tr\left( \vect{\Delta}(e^{j\omega}; \widetilde{\vect H},\vect{H})\vect{\Phi}_y(\omega)\vect{\Delta}^H(e^{j\omega}; \widetilde{\vect
  H},\vect{H})\right).
 \end{equation}
 In addition, let $\simeq$
denote an equality in which only dominating terms with respect to
$||\Htilde||$ are retained. Then, using (\ref{eq:MMSEeq1}), we
observe that
\begin{equation}\label{deltaapprox}
\begin{split}
 \vect{\Delta}(&\ejo; \Htilde,\vect{H}) = \vect{F}(\ejo;\vect{H}+\Htilde) - \vect{F}(\ejo;\vect{H}) \\
 &\simeq \lambda_x \Htilde^H \vect{\Phi}_y^{-1}  - \lambda_x^2 \vect{H}^H \vect{\Phi}_y^{-1} (\vect{H}
    \Htilde^H + \Htilde \vect{H}^H  ) \vect{\Phi}_y^{-1} \\
    = \lambda_x \Big( &\underbrace{\left( \vect{I} \!-\! \lambda_x \vect{H}^H \vect{\Phi}_y^{-1} \vect{H}
      \right)}_{= \vect{Q}} \Htilde^H \vect{\Phi}_y^{-1} -
      \lambda_x
    \vect{H}^H \vect{\Phi}_y^{-1} \Htilde \vect{H}^H
    \vect{\Phi}_y^{-1} \Big)
    \end{split}
\end{equation}
where we omitted the argument $\omega$ for simplicity. Inserting
\eqref{deltaapprox} in \eqref{jprim} results in the approximation
\begin{align}\label{traceapprox}
     J'(\omega; \Htilde,\vect{H}) & \simeq  \lambda_x^2 \tr\Bigl( \vect{Q} \Htilde^H \vect{\Phi}_y^{-1} \Htilde \vect{Q}
     \notag \\ & \qquad + \lambda_x^2 \left ( \vect{H}^H \vect{\Phi}_y^{-1}
  \Htilde \vect{H}^H \vect{\Phi}_y^{-1} \vect{H} \Htilde^H \vect{\Phi}_y^{-1} \vect{H} \right)\notag \\
   & \qquad -
  \lambda_x \vect{Q} \Htilde^H \vect{\Phi}_y^{-1} \vect{H} \Htilde^H \vect{\Phi}_y^{-1} \vect{H}
  \notag \\ & \qquad - \lambda_x \vect{H}^H \vect{\Phi}_y^{-1}
  \Htilde \vect{H}^H \vect{\Phi}_y^{-1} \Htilde \vect{Q}  \Bigr).
\end{align}
To rewrite this into a quadratic form in terms of
$\vecop(\Htilde)$ we use the facts that $\tr(\vect{A}\vect{B}) =
\tr(\vect{B}\vect{A})
=\vecop^T(\vect{A}^T)\vecop(\vect{B})=\vecop^H(\vect{A}^H)\vecop(\vect{B})$
and $\vecop(\vect{A}\vect{B}\vect{C})=(\vect{C}^T\otimes \vect{A})
\vecop(\vect{B})$ for matrices $\vect{A},\ \vect{B},$ and
$\vect{C}$ of compatible dimensions. Hence, we can rewrite
\eqref{traceapprox} as
\begin{align}\label{vecapprox}
&J'(\omega; \Htilde,\vect{H})  \simeq \vecop^H (\Htilde)
[\lambda_x^2 {\vect{Q}^2}^T\otimes \vect{\Phi}_y^{-1}]\vecop( \Htilde) \notag \\
& \quad + \vecop^H( \Htilde)[\lambda_x^4(\vect{H}^H
\vect{\Phi}_y^{-1} \vect{H})^T \otimes  \vect{\Phi}_y^{-1}
\vect{H}\vect{H}^H \vect{\Phi}_y^{-1}] \vecop(\Htilde)
  \notag \\
 &\quad   - \vecop^H(\Htilde)  [\lambda_x^3 (\vect{\Phi}_y^{-1}\vect{H}\vect{Q} )^T  \otimes \vect{\Phi}_y^{-1} \vect{H}] \vecop( \Htilde^H)
 \notag \\
  &\quad - \vecop^H (   \Htilde^H)  [\lambda_x^3(\vect{Q}\vect{H}^H\vect{\Phi}_y^{-1})^T  \otimes  \vect{H}^H \vect{\Phi}_y^{-1}] \vecop(\Htilde
  ).
\end{align}
In the next step,  we introduce the permutation matrix $\vect{\Pi}$
defined such that $\vecop(\Htilde^T) = \vect{\Pi}\vecop(\Htilde)$
for every $\Htilde$ to rewrite \eqref{vecapprox} as
\begin{align}\label{quadapprox}
&  J'(\omega; \Htilde,\vect{H})  \simeq
  \vecop^H (\Htilde) [\lambda_x^2{\vect{Q}^2}^T\otimes
\vect{\Phi}_y^{-1}]\vecop( \Htilde) \notag \\ & \quad + \vecop^H(
\Htilde)[\lambda_x^4(\vect{H}^H \vect{\Phi}_y^{-1} \vect{H})^T
\otimes \vect{\Phi}_y^{-1} \vect{H}\vect{H}^H \vect{\Phi}_y^{-1}]
\vecop(\Htilde)
  \notag \\
 &\quad   - \vecop^H(\Htilde)  [\lambda_x^3(\vect{\Phi}_y^{-1}\vect{H}\vect{Q} )^T  \otimes \vect{\Phi}_y^{-1} \vect{H}] \vect{\Pi}\vecop( \Htilde^*)
 \notag \\
  &\quad - \vecop^H (   \Htilde^*) \vect{\Pi}^T  [\lambda_x^3(\vect{Q}\vect{H}^H\vect{\Phi}_y^{-1})^T  \otimes  \vect{H}^H \vect{\Phi}_y^{-1}] \vecop(\Htilde
  ).
 \end{align}
We have now obtained a quadratic form. Note indeed that
the last two terms are just complex conjugates of each other and
thus we can write them as two times their real part.

\subsection{High SNR analysis}
In order to obtain a simpler expression for $\Ical_{\text{adm}}$,
we will assume high SNR in the data transmission phase. We
consider the practically relevant case where ${\rm
rank}\left(\vect{H}\right)=\min(n_T,n_R)$. Depending on the rank
of the channel matrix $\vect{H}$ we will have three different
cases:

\subsubsection*{Case 1: $\rank(\vect{H}) =n_R < n_T$}
Under this assumption, it can be shown that both the first and the
second terms on the right hand side of \eqref{quadapprox}
contribute to $\Ical_{\text{adm}}$. We have $\vect{Q}\to
\vect{\Pi}^\perp_{\vect{H}^H}$ and $\lambda_x \vect{\Phi}_y^{-1}
\to (\vect{H}\vect{H}^{H})^{-1}$ for high SNR. Here, and in what
follows, we use $\vect{\Pi}_\vect{X}=\vect{X}\vect{X}^\dagger$ to
denote the orthogonal projection matrix on the range-space of
$\vect{X}$ and $\vect{\Pi}_\vect{X}^\perp =
\vect{I}-\vect{\Pi}_\vect{X}$ to denote the projection on the
nullspace
of $\vect{X}^H$. Moreover, $\lambda_x
\vect{H}^H \vect{\Phi}_y^{-1} \vect{H} \to
\vect{\Pi}_{\vect{H}^H}$ and $\lambda_x^2 \vect{\Phi}_y^{-1}
\vect{H}\vect{H}^H \vect{\Phi}_y^{-1} \to (\vect{H}\vect{H}^{H}
)^{-1}$ for high SNR. As
$\vect{\Pi}^\perp_{\vect{H}^H}+\vect{\Pi}_{\vect{H}^H} =
\vect{I}$, summing the contributions from the first two terms in
\eqref{quadapprox} finally gives the high SNR approximation
\begin{equation}\label{IadmhighsnrCase2}
\Ical_{\text{adm}} = \lambda_x \vect{I}  \otimes
(\vect{H}\vect{H}^{H} )^{-1} .\end{equation}

\subsubsection*{Case 2: $\rank(\vect{H}) =n_R = n_T$}
For the non-singular channel case, the second term on the right
hand side of \eqref{quadapprox} dominates. Here, we have
 $\lambda_x \vect{H}^H \vect{\Phi}_y^{-1} \vect{H} \to \vect{I}$ and
 $\lambda_x^2 \vect{\Phi}_y^{-1} \vect{H}\vect{H}^H \vect{\Phi}_y^{-1} \to (\vect{H}\vect{H}^H)^{-1}$ for high SNR. Clearly, this results in the same expression for $\Ical_{\text{adm}}$ as in Case 1, namely,
\begin{equation}\label{IadmhighsnrCase4}
\Ical_{\text{adm}} = \lambda_x \vect{I}  \otimes
(\vect{H}\vect{H}^{H} )^{-1} .\end{equation}

\subsubsection*{Case 3:  $\rank(\vect{H}) =n_T < n_R$}

In this case, the second term on the right hand side of
\eqref{quadapprox} dominates. When $\rank(\vect{H}) =n_T$ we get
 $\lambda_x \vect{H}^H \vect{\Phi}_y^{-1} \vect{H} \to \vect{I}$ and $\lambda_x^2 \vect{\Phi}_y^{-1} \vect{H}\vect{H}^H \vect{\Phi}_y^{-1} \to \vect{\Phi}_n^{-1/2}[\vect{\Phi}_n^{-1/2}\vect{H}\vect{H}^H\vect{\Phi}_n^{-1/2}]^\dagger \vect{\Phi}_n^{-1/2}
$ for high SNR.
Using these approximations finally gives the high SNR approximation
\begin{equation}\label{IadmhighsnrCase3}
\Ical_{\text{adm}} =\lambda_x \vect{I}  \otimes \left(\frac{1}{2\pi}\int_{-\pi}^\pi \vect{\Phi}_n^{-1/2}[\vect{\Phi}_n^{-1/2}\vect{H}\vect{H}^H\vect{\Phi}_n^{-1/2}]^\dagger \vect{\Phi}_n^{-1/2}  \ d\omega\right).\nonumber\\
\end{equation}

\subsection{Low SNR analysis}
For the low SNR regime, we do not need to differentiate our
analysis for the cases $n_T\geq n_R$ and $n_T<n_R$, because now
$\vect{\Phi}_y\rightarrow \vect{\Phi}_n$. It can be shown that the
first term on the right hand side of \eqref{quadapprox} dominates;
that is, the term involving
$$  \lambda_x^2 \ \left ( (\vect{Q}^{2})^T \otimes \vect{\Phi}_y^{-1}
\right ).
$$
Moreover, $\vect{Q}\rightarrow \vect{I}$ and
$\vect{\Phi}^{-1}_y\rightarrow\vect{\Phi}_n^{-1}$. This yields
 \begin{equation}
 \Ical_{\text{adm}}= \vect{I}
\otimes \left ( \frac{\lambda_x^2}{2\pi} \int_{-\pi}^\pi
\vect{\Phi}_n^{-1} \ d\omega \right ).
\end{equation}

\section{Proof of Theorem~\ref{teo:1}} \label{app:teo_1}

For the proof of Theorem~\ref{teo:1}, we require some preliminary
results. Lemma~\ref{lem:1} and Lemma~\ref{lem:2} will be used to
establish the uniqueness part of Theorem~\ref{teo:1}, and
Lemma~\ref{lem:3} is an extension of a standard result in
majorization theory, which is used in the main part of the proof.

\vspace{5pt} %
\begin{lem} \label{lem:1} %
Let $\vect{D} \in \mathbb{R}^{n \times n}$ be a diagonal matrix
with elements $d_{1,1} > \cdots > d_{n,n} > 0$. If $\vect{U} \in
\mathbb{C}^{n \times n}$ is a unitary matrix such that $\vect{U}
\vect{D} \vect{U}^H$ has diagonal $(d_{1,1}, \ldots, d_{n,n})$,
then $\vect{U}$ is of the form $\vect{U} = \diag(u_{1,1}, \ldots,
u_{n,n})$, where $|u_{i,i}| = 1$ for $i = 1, \ldots, n$. This also
implies that $\vect{U} \vect{D} \vect{U}^H = \vect{D}$.
\end{lem}

\begin{proof}
Let $\vect{V} = \vect{U} \vect{D} \vect{U}^H$. The equation for
$(\vect{V})_{i,i}$ is
\begin{align*}
\sum_{k = 1}^n d_{k,k} |u_{i,k}|^2 = d_{i,i}
\end{align*}
from which we have, by the orthonormality of the columns of
$\vect{U}$, that
\begin{align} \label{eq:lem1_1}
\sum_{k = 1}^n \frac{d_{k,k}}{d_{i,i}} |u_{i,k}|^2 = 1 = \sum_{k =
1}^n |u_{i,k}|^2.
\end{align}
We now proceed by induction on $i = 1, \ldots, n$ to show that the
$i$th column of $\vect{U}$ is $[0 \; \cdots \; 0 \; u_{i,i} \; 0
\; \cdots \; 0]^T$ with $|u_{i,i}| = 1$. For $i=1$, it follows
from \eqref{eq:lem1_1} and the fact that $\vect{U}$ is unitary
that
\begin{eqnarray*}
&&|u_{1,1}|^2 + \left| \frac{d_{2,2}}{d_{1,1}} u_{2,1} \right|^2 +
\cdots + \left| \frac{d_{n,n}}{d_{1,1}} u_{n,1} \right|^2 \\
&&\qquad \qquad \qquad \qquad \qquad = |u_{1,1}|^2 + \cdots + |u_{n,1}|^2= 1.
\end{eqnarray*}
However, since $d_{1,1} > \cdots > d_{n,n} > 0$, the only way to
satisfy this equation is to have $|u_{1,1}| = 1$ and $u_{i,1} = 0$
for $i = 2, \ldots, n$. Now, if the assertion holds for $i = 1,
\ldots, k$, the orthogonality of the columns of $\vect{U}$ implies
that $u_{i, k + 1} = 0$ for $i = 1, \ldots, k$, and by following a
similar reasoning as for the case $i = 1$ we deduce that $|u_{k +
1, k + 1}| = 1$ and $u_{i, k + 1} = 0$ for $i = k + 2, \ldots, n$.
\end{proof}

\vspace{5pt} %
\begin{lem} \label{lem:2} %
Let $\vect{D} \in \mathbb{R}^{N \times N}$ be a diagonal matrix
with elements $d_{1,1} > \cdots > d_{N,N} > 0$. If $\vect{U} \in
\mathbb{C}^{N \times n}$, with $n \leq N$, is such that
$\vect{U}^H \vect{U} = I$ and $\vect{V} = \widetilde{\vect{D}}
\vect{U} \widetilde{\vect{D}}^{-1}$ (where $\widetilde{\vect{D}} =
\diag(d_{1,1}, \ldots, d_{n,n})$) also satisfies $\vect{V}^H
\vect{V} = \vect{I}$, then $\vect{U}$ is of the form $\vect{U} =
[\diag(u_{1,1}, \ldots, u_{n,n}) \quad \vect{0}_{N - m,n}]^T$, where
$|u_{i,i}| = 1$ for $i = 1, \ldots, n$.
\end{lem}

\begin{proof}
The idea is similar to the proof of Lemma~\ref{lem:1}. We proceed
by induction on the $i$th column of $\vect{V}$. For the first
column of $\vect{V}$ we have, by the orthonormality of the columns
of $\vect{U}$ and $\vect{V}$, that
\begin{align*}
&|u_{1,1}|^2 + \left| \frac{d_{2,2}}{d_{1,1}} u_{2,1} \right|^2 +
\cdots + \left| \frac{d_{N,N}}{d_{1,1}} u_{N,1} \right|^2 \\
&\qquad \qquad \qquad \qquad \qquad = 1 \\
&\qquad \qquad \qquad \qquad \qquad = |u_{1,1}|^2 + \cdots +
|u_{N,1}|^2.
\end{align*}
Since $d_{1,1} > \cdots > d_{N,N} > 0$, the only way to satisfy
this equation is to have $|u_{1,1}| = 1$ and $u_{i,1} = 0$ for $i
= 2, \ldots, N$. If now the assertion holds for columns $1$ to
$k$, the orthogonality of the columns of $\vect{U}$ implies that
$u_{i, k + 1} = 0$ for $i = 1, \ldots, k$, and by following a
similar reasoning as for the first column of $\vect{U}$ we have that
$|u_{k + 1, k + 1}| = 1$ and $u_{i, k + 1} = 0$ for $i = k + 2,
\ldots, N$.
\end{proof}

\vspace{5pt} %
\begin{lem} \label{lem:3}
Let $\vect{A}, \vect{B} \in \mathbb{C}^{n \times n}$ be Hermitian
matrices. Arrange the eigenvalues $a_1, \ldots, a_n$ of $\vect{A}$
in a descending order, and the eigenvalues $b_1, \ldots, b_n$ of
$\vect{B}$ in an ascending order. Then $\tr ( \vect{A} \vect{B}
)\geq \sum_{i = 1}^n a_i b_i$. Furthermore, if $\vect{B} =
\diag(b_1, \ldots, b_n)$ and both matrices have distinct
eigenvalues, then $\tr ( \vect{A} \vect{B} ) = \sum_{i = 1}^n a_i
b_i$ if and only if $\vect{A} = \diag(a_1, \ldots, a_n)$.
\end{lem}

\begin{proof}
See \cite[Theorem~9.H.1.h]{MarshallO:79} for the proof of the
first assertion. For the second part, notice that if $\vect{B} =
\diag(b_1, \ldots, b_n)$, then by
\cite[Theorem~6.A.3]{MarshallO:79}
\begin{align*}
\tr ( \vect{A} \vect{B} ) = \sum_{i = 1}^n (\vect{A})_{i,i} b_i
\geq \sum_{i = 1}^n (\vect{A})_{[i,i]} b_i
\end{align*}
where $\{(\vect{A})_{[i,i]} \}_{i = 1, \ldots, n}$ denotes the
ordered set $\{ (\vect{A})_{1,1}, \ldots, (\vect{A})_{n,n} \}$
sorted in descending order. Since $\{(\vect{A})_{[i,i]} \}_{i = 1,
\ldots, n}$ is \emph{majorized} by $\{a_1, \ldots, a_n\}$, and the
$b_i$'s are distinct, we can use
\cite[Theorem~3.A.2]{MarshallO:79} to show that
\begin{align*}
\sum_{i = 1}^n (\vect{A})_{[i,i]} b_i > \sum_{i = 1}^n a_i b_i
\end{align*}
unless $(\vect{A})_{[i,i]} = a_i$ for every $i = 1, \ldots, n$.
Therefore, $\tr( \vect{A} \vect{B} ) = \sum_{i = 1}^n a_i b_i$ if
and only if the diagonal of $\vect{A}$ is $(a_1, \ldots, a_n)$.
Now we have to prove that $\vect{A}$ is actually diagonal, but
this follows from Lemma~\ref{lem:1}.
\end{proof}

\vspace{5pt} %
\textbf{\emph{Proof of Theorem~\ref{teo:1}.}} First, we simplify
the expressions in \eqref{eq:teo1_1}. Using the eigendecompositions in \eqref{eq:teo1_3} of $\vect{A}$ and $\vect{B}$,
we see that
\begin{align*}
\vect{P} \vect{A}^{-1} \vect{P}^H  \succeq \vect{B} \quad
&\Leftrightarrow \quad \vect{P} \vect{U}_A \vect{D}_A^{-1} \vect{U}_A^H \vect{P}^H  \succeq \vect{U}_B \vect{D}_B \vect{U}_B^H  \\
&\Leftrightarrow \quad \vect{U}_B^H \vect{P} \vect{U}_A
\vect{D}_A^{-1} \vect{U}_A^H \vect{P}^H \vect{U}_B \succeq
\vect{D}_B.
\end{align*}
Now, define $\bar{\vect{P}} = \vect{U}_B^H \vect{P} \vect{U}_A
\vect{D}_A^{-1/2}$ and observe that
\begin{eqnarray*}
\tr ( \vect{P} \vect{P}^H )%
&=& \tr  \left((\vect{U}_B \bar{\vect{P}} \vect{D}_A^{-H/2} \vect{U}_A^H ) (\vect{U}_B \bar{\vect{P}} \vect{D}_A^{-H/2} \vect{U}_A^H)^H \right) \\ %
&=&\tr  (\vect{U}_B \bar{\vect{P}} \vect{D}_A^{-1} \bar{\vect{P}}^H  \vect{U}_B^H )= \tr (\bar{\vect{P}}^H \bar{\vect{P}} \vect{D}_A^{-1}).
\end{eqnarray*}
Therefore, \eqref{eq:teo1_1} is equivalent to
\begin{align} \label{eq:teo_1_4}
\begin{array}{cl}
\minimize{\vect{P} \in \mathbb{C}^{n \times N}} & \tr (\bar{\vect{P}}^H \bar{\vect{P}} \vect{D}_A^{-1}) \\ %
\text{s.t.}                                         & \bar{\vect{P}} \bar{\vect{P}}^H \succeq \vect{D}_B. %
\end{array}
\end{align}
To further simplify our problem, consider the singular value
decomposition $\bar{\vect{P}} = \vect{U} \vect{\Sigma}
\vect{V}^H$, where $\vect{U} \in \mathbb{C}^{n \times n}$ and
$\vect{V} \in \mathbb{C}^{N \times N}$ are unitary matrices and
$\vect{\Sigma}$ has the structure
\begin{align*}
\vect{\Sigma} =\! \left[ \!\!\begin{array}{cccccc}
\sigma_1 &        & 0        & 0      & \cdots & 0      \\
         & \ddots &          & \vdots &        & \vdots \\
0        &        & \sigma_m & 0      & \cdots & 0
\end{array} \!\! \right] \textrm{or}\,\,
\vect{\Sigma} =\! \left[ \!\!\begin{array}{ccc}
\sigma_1 &        & 0          \\
         & \ddots &           \\
0        &        & \sigma_m  \\
0      & \cdots & 0 \\
\vdots &        & \vdots\\
0      & \cdots & 0
\end{array} \!\! \right]
\end{align*}
depending on whether $N \geq n$ or $N<n$. The singular values are
ordered such that $\sigma_1 \geq \cdots \geq \sigma_m > 0$. Now,
observe that \eqref{eq:teo_1_4} is equivalent to
\begin{align} \label{eq:teo_1_5}
\begin{array}{cl}
\minimize{\vect{P} \in \mathbb{C}^{n \times N}} & \tr ( \vect{V}^H \vect{\Sigma}^H \vect{\Sigma} \vect{V}^H \vect{D}_A^{-1}) \\ %
\text{s.t.}                                         & \vect{U} \vect{\Sigma}  \vect{\Sigma}^H \vect{U}^H \succeq \vect{D}_B. %
\end{array}
\end{align}
With this problem formulation, it follows (from Sylvester's law
of inertia \cite{Ostrowski:60}) that we need $m \geq
\rank(\vect{D}_B)$ to achieve feasibility in the constraint (i.e.,
having at least as many non-zero singular values of
$\vect{\Sigma}$ as non-zero eigenvalues in $\vect{D}_B$). This
corresponds to the condition $N \geq \rank(\vect{B})$ in the
theorem.

Now we will show that $\vect{U}$ and $\vect{V}$ can be taken to be
the identity matrices. Using Lemma \ref{lem:3}, the cost function
can be lower bounded as
\begin{equation} \label{eq:teo_1_6}
\begin{split} \tr ( \vect{V} \vect{\Sigma}^H \vect{\Sigma} \vect{V}^H
\vect{D}_A^{-1}) &\geq \sum_{j=1}^{n} \lambda_{n-j+1}(\vect{D}_A)
\lambda_{j}(\vect{V} \vect{\Sigma}^H \vect{\Sigma} \vect{V}^H) \\
&= \sum_{j=1}^{m} (\vect{D}_A)_{jj} \sigma_j^2
\end{split}
\end{equation}
where $\lambda_j(\cdot)$ denotes the $j$th largest eigenvalue. The
equality is achieved if $\vect{V}=\vect{I}$, and observe that we
can select $\vect{V}$ in this manner without affecting the
constraint.

To show that $\vect{U}$ can also be taken as the identity matrix,
notice that the cost function in \eqref{eq:teo_1_5} does not
depend on $\vect{U}$, while the constraint implies (by looking at
the diagonal elements of the inequality and recalling that
$\vect{U}$ is unitary) that
\begin{align} \label{eq:teo_1_9}
\sigma_i^2 \geq (\vect{D}_B)_{i, i}, \quad \quad i = 1, \ldots, m,
\end{align}
requiring $m \geq \rank(\vect{D}_B)$. Suppose that
$\bar{\vect{U}}$ and $\bar{\vect{\Sigma}}$ minimize the cost.
Then, we can replace $\bar{\vect{U}}$ by $\vect{I}$ and
 satisfy the constraint, without affecting the cost in
 \eqref{eq:teo_1_6}. This means that there exists an optimal solution with $\vect{U} = \vect{I}$.

With $\vect{U} = \vect{I}$ and $\vect{V} = \vect{I}$, the problem
\eqref{eq:teo_1_5} is equivalent (in terms of $\vect{\Sigma}$) to
\begin{align*}
\begin{array}{cl}
\minimize{\sigma_1\geq0,\ldots,\sigma_m \geq 0} & \sum_{i = 1}^m \sigma_i^2 (\vect{D}_A)_{i, i} \\ %
\text{s.t.}                                         & \sigma_i^2 \geq (\vect{D}_B)_{i, i}, \quad i = 1, \ldots, m. %
\end{array}
\end{align*}
It is easy to see that the optimal solution for this problem is
$
\sigma_i^{\textrm{opt}} = \sqrt{(\vect{D}_B)_{i, i}},
i = 1, \ldots, m.
$
By creating an optimal $\vect{\Sigma}$, denoted as
$\vect{\Sigma}^{\textrm{opt}}$, with the singular values
$\sigma_1^{\textrm{opt}},\ldots,\sigma_m^{\textrm{opt}}$, we
achieve an optimal solution
\begin{eqnarray*}
\vect{P}^{\textrm{opt}}
= \vect{U}_B \bar{\vect{P}} \vect{D}_A^{1/2}\vect{U}_A^H
= \vect{U}_B \vect{\Sigma}^{\textrm{opt}} \vect{D}_A^{1/2}\vect{U}_A^H
 = \vect{U}_B
\vect{D}_P \vect{U}_A^H\nonumber
\end{eqnarray*}
with $\vect{D}_P$ as stated in the theorem.

Finally, we will show how to characterize all optimal solutions
for the case when $\vect{A}$ and $\vect{B}$ have distinct non-zero
eigenvalues (thus, $m=n$). The optimal solutions need to give
equality in \eqref{eq:teo_1_6} and thus Lemma~\ref{lem:3} gives
that $\vect{V} \vect{\Sigma} \vect{\Sigma}^H \vect{V}^H$ is
diagonal and equal to $\vect{\Sigma} \vect{\Sigma}^H$.
Lemma~\ref{lem:1} then implies that $\vect{V} = \diag(v_{1, 1},
\ldots, v_{n, n})$ with $|v_{i, i}| = 1$ for $i = 1, \ldots, n$.

For the optimal $\vect{\Sigma}$, we have that $\sigma_i^2 =
(\vect{D}_B)_{i, i}$ for $i = 1, \ldots, n$, so the diagonal
elements of $\vect{U} \vect{\Sigma}  \vect{\Sigma}^H \vect{U}^H -
\vect{D}_B$ are zero. Since $\vect{U} \vect{\Sigma}
\vect{\Sigma}^H \vect{U}^H - \vect{D}_B \succeq 0$ for every
feasible solution of \eqref{eq:teo_1_5}, $\vect{U}$ has to satisfy
$\vect{U} \vect{\Sigma}  \vect{\Sigma}^H \vect{U}^H = \vect{D}_B$.
Lemma~\ref{lem:2} then establishes that the first $n$ columns of
$\vect{U}$ are of the form $[\diag(u_{1, 1}, \ldots, u_{n, n})
\quad \vect{0}_{N - m, n}]^T$, where $|u_{i,i}| = 1$ for $i = 1,
\ldots, n$. Since $\vect{U}$ has to be unitary, and its last $N -
n + 1$ columns play no role in $\bar{\vect{P}}$ (due to the form
of $\vect{\Sigma}$), we can take them as $[\vect{0}_{n, N - m + 1}
\quad \vect{I}_{N - m + 1}]^T$ without loss of generality.

Summarizing, an optimal solution is given by \eqref{eq:teo1_3}.
When $\vect{A}$ and $\vect{B}$ have distinct eigenvalues,
$\vect{V}$ and $\vect{U}$ can only multiply the columns of
$\vect{U}_A$ and $\vect{U}_B$, respectively, by complex scalars of
unit magnitude.

\section{Proof of Theorem~\ref{teo:2} and Theorem~\ref{teo:3}} \label{app:cor_2}

Before proving Theorem~\ref{teo:2} and \ref{teo:3}, a lemma will
be given that characterizes equivalences between different sets of
feasible training matrices $\vect{P}$.

\begin{lem} \label{lem:4}
Let $\vect{B} \in \mathbb{C}^{n \times n}$ and $\vect{C} \in
\mathbb{C}^{m \times m}$ be Hermitian matrices, and $f:
\mathbb{C}^{n \times N} \rightarrow \mathbb{C}^{n \times n}$ be
such that $f(\vect{P})=f(\vect{P})^H$. Then, the following sets are equivalent
\begin{equation}\{\vect{P} | f(\vect{P}) \otimes \vect{I}
\succeq \vect{B} \otimes \vect{C}\} = \{\vect{P} | f(\vect{P})
\succeq \lambda_{\textrm{max}}(\vect{C}) \vect{B} \}.
\end{equation}
\end{lem}
\begin{proof}
The equivalence will be proved by showing that the left hand side
(LHS) is a subset of right hand side (RHS), and \textit{vice
versa}. First, assume that $f(\vect{P}) \succeq
\lambda_{\textrm{max}}(\vect{C}) \vect{B}$, then
\begin{equation}
\begin{split}
f(\vect{P}) \otimes \vect{I} &\succeq
\lambda_{\textrm{max}}(\vect{C}) \vect{B} \otimes \vect{I}
\\ & =
 (\vect{B} \otimes \lambda_{\textrm{max}}(\vect{C}) \vect{I}) \succeq
 (\vect{B} \otimes \vect{C}).
 \end{split}
\end{equation}
Hence, $\textrm{RHS} \subseteq \textrm{LHS}$.

Next, assume that $f(\vect{P}) \otimes \vect{I} \succeq \vect{B}
\otimes \vect{C}$, but for the purpose of contradiction that
$f(\vect{P}) \not \succeq \lambda_{\textrm{max}}(\vect{C})
\vect{B}$. Then, there exists a vector $\vect{x}$ such that
$\vect{x}^H (f(\vect{P}) - \lambda_{\textrm{max}}(\vect{C})
\vect{B}) \vect{x} < 0$. Let $\vect{v}$ be an eigenvector of
$\vect{C}$ that corresponds to $\lambda_{\textrm{max}}(\vect{C})$
and define $\vect{y} = \vect{x} \otimes \vect{v}$. Then
\begin{equation}
\begin{split}
\vect{y} (f(\vect{P})& \otimes \vect{I} - \vect{B} \otimes
\vect{C}) \vect{y} \\ &= (\vect{x}^H f(\vect{P}) \vect{x}) \|
\vect{v} \|^2 -
 (\vect{x}^H \vect{B}\vect{x}) (\vect{v}^H \vect{C} \vect{v}) \\ &= \vect{x}^H (f(\vect{P}) - \lambda_{\textrm{max}}(\vect{C}) \vect{B} )\vect{x} \|
\vect{v} \|^2 < 0
\end{split}
\end{equation}
which is a contradiction. Hence, $\textrm{LHS} \subseteq
\textrm{RHS}$.
\end{proof}

\textbf{\emph{Proof of Theorem~\ref{teo:2}.}} Rewrite the
constraint as
\begin{equation}
\begin{split}
&\vect{\Ptilde}^H (\vect{S}_{Q}^T \otimes \vect{S}_{R})^{-1}
\vect{\Ptilde}  \succeq c\Ical_T^T \otimes \Ical_R \\
 \Leftrightarrow \quad & (\vect{P} \vect{S}_{Q}^{-1} \vect{P}^H)^T
\otimes \vect{S}_{R}^{-1} \succeq c\Ical_T^T \otimes \Ical_R \\
 \Leftrightarrow \quad & (\vect{P} \vect{S}_{Q}^{-1} \vect{P}^H)
\otimes \vect{I} \succeq c\Ical_T \otimes \vect{S}_{R} \Ical_R.
\end{split}
\end{equation}
Let $f(\vect{P}) = \vect{P} \vect{S}_{Q}^{-1} \vect{P}^H$. Then
Lemma \ref{lem:4} gives that the set of feasible $\vect{P}$ is
equivalent to the set of feasible $\vect{P}$ with the constraint
\begin{equation}
(\vect{P} \vect{S}_{Q}^{-1} \vect{P}^H) \succeq c
\lambda_{\textrm{max}}(\vect{S}_{R} \Ical_R) \Ical_T .
\end{equation}

\textbf{\emph{Proof of Theorem~\ref{teo:3}.}}

In the case that $\vect{R}_R=\vect{S}_R$, the constraint can be
rewritten as
\begin{equation}
(\vect{P} \vect{S}_{Q}^{-1} \vect{P}^H + \vect{R}_{T}^{-1})^T
\otimes \vect{I}   \succeq c \Ical_T^T \otimes \vect{S}_{R}
\Ical_R.
\end{equation}
With $f(\vect{P}) = \vect{P} \vect{S}_{Q}^{-1} \vect{P}^H +
\vect{R}_{T}^{-1}$, Lemma \ref{lem:4} can be applied to achieve
the equivalent constraint
\begin{equation}
\begin{split}
&\vect{P} \vect{S}_{Q}^{-1} \vect{P}^H + \vect{R}_{T}^{-1}
 \succeq c\lambda_{\textrm{max}}(\vect{S}_{R} \Ical_R)
 \Ical_T \\
 \Leftrightarrow &\quad \vect{P} \vect{S}_{Q}^{-1} \vect{P}^H
 \succeq c\lambda_{\textrm{max}}(\vect{S}_{R} \Ical_R)
 \Ical_T-\vect{R}_{T}^{-1} \\
 \Leftrightarrow &\quad \vect{P} \vect{S}_{Q}^{-1} \vect{P}^H
 \succeq [c\lambda_{\textrm{max}}(\vect{S}_{R} \Ical_R)
 \Ical_T-\vect{R}_{T}^{-1}]_+
\end{split}
\end{equation}
where the last equality follows from the fact that the left hand side is
positive semi-definite.

In the case that $\vect{R}_R^{-1}=\Ical_R$, the constraint can be
rewritten as
\begin{equation}
\begin{split}
&(\vect{P} \vect{S}_{Q}^{-1} \vect{P}^H)^T \otimes
\vect{S}_{R}^{-1}
\succeq (c\Ical_T-\vect{R}_{T})^T \otimes \Ical_R \\
\Leftrightarrow &\quad (\vect{P} \vect{S}_{Q}^{-1} \vect{P}^H)^T
\otimes \vect{S}_{R}^{-1} \succeq [c\Ical_T-\vect{R}_{T}]_+^T
\otimes
 \Ical_R.
\end{split}
\end{equation}
Observe that this expression is identical to the constraint in
(\ref{eq:teo2_1}), except that the positive semi-definite
$\Ical_T$ has been replaced by $ [c\Ical_T-\vect{R}_{T}]_+$. Thus,
the equivalence follows directly from Theorem \ref{teo:2}.

In the case $\vect{R}_T^{-1}=\Ical_T$, the constraint can be
rewritten as
\begin{equation}
\begin{split}
&(\vect{P} \vect{S}_{Q}^{-1} \vect{P}^H)^T \otimes
\vect{S}_{R}^{-1} \succeq
\Ical_T^T \otimes (c\Ical_R-\vect{R}_{R}) \\
\Leftrightarrow &\quad (\vect{P} \vect{S}_{Q}^{-1} \vect{P}^H)^T
\otimes \vect{S}_{R}^{-1} \succeq \Ical_T^T \otimes
[c\Ical_R-\vect{R}_{R}]_+.
\end{split}
\end{equation}
As in the previous case, the equivalence follows directly from
Theorem \ref{teo:2}. 

\section{Proof of Theorem~\ref{teo:4}} \label{app:cor_3}

Our basic assumption is that $\Ical_T,\Ical_R$ are both Hermitian matrices, which
is encountered in the applications presented in this paper. Denoting by $\vect{P}'$ the matrix $\vect{P}^{T}$
and using the fact that\footnote{For a Hermitian positive semidefinite matrix $\vect{A}$, we consider here that $\vect{A}^{1/2}$ is the matrix with the same eigenvectors as $\vect{A}$ and eigenvalues the square roots of the corresponding eigenvalues of $\vect{A}$. With this definition of the square root of a Hermitian positive semidefinite matrix, it is clear that $\vect{A}^{1/2}=\vect{A}^{H/2}$, leading to $\vect{A}=\vect{A}^{1/2}\vect{A}^{H/2}=\vect{A}^{H/2}\vect{A}^{1/2}$.} $\Ical_{\text{adm}}=\left(\Ical_T' \otimes \Ical_R\right)^{1/2}\left(\Ical_T' \otimes \Ical_R\right)^{1/2}$,
it can be seen that our optimization problem takes the following form
\begin{align} \label{eq:teo4_2}
\begin{array}{cl}
\minimize{\vect{P}' \in \mathbb{C}^{B \times n_T}} & J(\vect{H}) \\ %
\text{s.t.}                                 &{\mathrm{tr}(\vect{P}'\vect{P'}^{H})}\leq \mathcal{P}  %
\end{array}
\end{align}
where $J(\vect{H})=\E_{\widetilde{\vect{H}}}\left\{J(\widetilde{\vect{H}}, \vect{H})\right\}$ is given by the expression
\begin{eqnarray*}
&&\mathrm{tr}\left\{\left[\Ical_T'^{-1/2}\vect{P'}^{H}\vect{S'}_Q^{-1}\vect{P'}\Ical_T'^{-1/2}\otimes \Ical_R^{-1/2}\vect{S}_R^{-1}\Ical_R^{-1/2}\right]^{-1}\right\}\\
&&=\mathrm{tr}\left\{\left[\Ical_T'^{-1/2}\vect{P'}^{H}\vect{S'}_Q^{-1}\vect{P'}\Ical_T'^{-1/2}\right]^{-1}\otimes \Ical_R^{1/2}\vect{S}_R\Ical_R^{1/2}\right\}.
\end{eqnarray*}
Using the fact that $\mathrm{tr}\left(\vect{A}\otimes\vect{B}\right)=\mathrm{tr}\left(\vect{A}\right)\mathrm{tr}\left(\vect{B}\right)$ for square matrices $\vect{A}$ and $\vect{B}$, it is clear from the last expression that the optimal training matrix can be found by minimizing
\begin{eqnarray}
\mathrm{tr}\left\{\left[\vect{V}_T^H\Ical_T'^{-1/2}\vect{P'}^{H}\vect{S'}_Q^{-1}\vect{P'}\Ical_T'^{-1/2}\vect{V}_T\right]^{-1}\right\},
\end{eqnarray}
where $\vect{V}_T$ denotes the modal matrix of $\Ical_T'$ corresponding to an arbitrary ordering of its eigenvalues. Here, we have used the invariance of the trace operator under unitary transformations.
First, note that for an arbitrary Hermitian positive definite matrix $\vect{A}$, $\mathrm{tr}\left(\vect{A}^{-1}\right)=\sum_i 1/\lambda_i\left(\vect{A}\right)$, where $\lambda_i\left(\vect{A}\right)$ is the $i$th eigenvalue of $\vect{A}$. Since the function $1/x$ is strictly convex
for $x>0$, $\mathrm{tr}\left(\vect{A}^{-1}\right)$ is a Schur-convex function with respect to the eigenvalues of $\vect{A}$~\cite{MarshallO:79}. Additionally, for any Hermitian matrix $\vect{A}$, the vector of its diagonal entries is majorized by the vector of its eigenvalues~\cite{MarshallO:79}. Combining the last two results, it follows that $\mathrm{tr}\left(\vect{A}^{-1}\right)$ is minimized when $\vect{A}$ is diagonal. Therefore, we may choose the modal matrices of $\vect{P}'$ in such a way
that $\vect{V}_T^H\Ical_T'^{-1/2}\vect{P'}^{H}\vect{S'}_Q^{-1}\vect{P'}\Ical_T'^{-1/2}\vect{V}_T$ is diagonalized. Suppose that the
singular value decomposition (SVD) of $\vect{P'}^{H}$ is $\vect{U}\vect{D}_{P'}\vect{V}^{H}$ and that the modal
matrix of $\vect{S}_{Q}'$, corresponding to arbitrary ordering of its eigenvalues, is $\vect{V}_{Q}$. Setting $\vect{U}=\vect{V}_T$ and $\vect{V}=\vect{V}_{Q}$, $\vect{V}_T^H\Ical_T'^{-1/2}\vect{P'}^{H}\vect{S'}_Q^{-1}\vect{P'}\Ical_T'^{-1/2}\vect{V}_T$ is diagonalized and is given by the expression
\[
\vect{\Lambda}_T^{-1/2}\vect{D}_{P'}\vect{\Lambda}_Q^{-1}\vect{D}_{P'}\vect{\Lambda}_T^{-1/2}.
\]
Here, $\vect{\Lambda}_T$ and $\vect{\Lambda}_Q$ are the diagonal eigenvalue matrices containing the eigenvalues of $\Ical_T'$ and $\vect{S'}_Q$, respectively, in their
main diagonals. The ordering of the eigenvalues corresponds to $\vect{V}_T$ and
$\vect{V}_{Q}$. Clearly, by reordering the columns of $\vect{V}_T$ and
$\vect{V}_{Q}$, we can reorder the eigenvalues in $\vect{\Lambda}_T$ and $\vect{\Lambda}_Q$. Assume that there are two
different permutations $\pi, \varpi$ such that $\pi\left((\vect{\Lambda}_T)_{1,1}\right),\ldots,\pi\left((\vect{\Lambda}_T)_{n_T,n_T}\right)$ and $\varpi\left((\vect{\Lambda}_Q)_{1,1}\right),\ldots,\varpi\left((\vect{\Lambda}_Q)_{B,B}\right)$ minimize $J(\vect{H})$ subject to our training energy constraint. Then,
 the entries of the corresponding eigenvalue matrix of $\vect{V}_T^H\Ical_T'^{-1/2}\vect{P'}^{H}\vect{S'}_Q^{-1}\vect{P'}\Ical_T'^{-1/2}\vect{V}_T$ are
 $(\vect{D}_{P'})_{i,i}^{2}/\left(\pi\left((\vect{\Lambda}_T)_{i,i}\right)\varpi\left((\vect{\Lambda}_Q)_{i,i}\right)\right),i=1,2,\ldots, n_T$ ($B\geq n_T$).
 Setting $(\vect{D}_{P'})_{i,i}^{2}=\kappa_{i}, i=1,2,\ldots,n_T$, the optimization problem (\ref{eq:teo4_2}) results in
 \begin{align} \label{eq:teo4_3}
\begin{array}{cl}
\minimize{\pi,\varpi,\kappa_{i},i=1,2,\ldots, n_T} & \sum_{i=1}^{n_T}\frac{1}{\frac{\kappa_{i}}{\pi\left((\vect{\Lambda}_T)_{i,i}\right)\varpi\left((\vect{\Lambda}_Q)_{i,i}\right)}} \\ %
\text{s.t.}                                 &\sum_{i=1}^{n_T}\kappa_{i}\leq \mathcal{P}  %
\end{array}
\end{align}
which leads to
 \begin{align} \label{eq:teo4_4}
\begin{array}{cl}
\minimize{\pi,\varpi,\kappa_{i},i=1,2,\ldots, n_T} & \sum_{i=1}^{n_T}\frac{\alpha_i}{\kappa_i} \\ %
\text{s.t.}                                 &\sum_{i=1}^{n_T}\kappa_{i}\leq \mathcal{P}  %
\end{array}
\end{align}
where $\alpha_i=\pi\left((\vect{\Lambda}_T)_{i,i}\right)\varpi\left((\vect{\Lambda}_Q)_{i,i}\right),i=1,2,\ldots, n_T$. Forming the Lagrangian
of the last problem, it can be seen that
\[
(\vect{D}_{P'})_{i,i}=\sqrt{\frac{\mathcal{P}\sqrt{\alpha_i}}{\sum_{j=1}^{n_T}\sqrt{\alpha_j}}}, i=1,2,\ldots, n_T
\]
while the objective value equals to $\left(\sum_{i=1}^{n_T}\sqrt{\alpha_i}\right)^{2}/\mathcal{P}$. Using Lemma~\ref{lem:3}, it can be seen
that $\pi$ and $\varpi$ should correspond to opposite orderings of $(\vect{\Lambda}_T)_{i,i},(\vect{\Lambda}_Q)_{j,j}, i=1,2,\ldots, n_T, j=1,2,\ldots, B$, respectively. Since $B$ can be greater than $n_T$, the eigenvalues of $\Ical_T'$ must be set in decreasing order and those of $\vect{S'}_Q$ in increasing order.

\section{Proof of Theorem~\ref{teo:5}} \label{app:cor_4}

Using the factorization $\Ical_{\text{adm}}=\left(\Ical_T' \otimes \Ical_R\right)^{1/2}\left(\Ical_T' \otimes \Ical_R\right)^{1/2}$, we can see that $E\left\{J(\vect{\tilde{H}},\vect{H})\right\}$ is given by the expression
\begin{eqnarray}\label{eq:detMMSE-1}
&\mathrm{tr}&\left\{\left[\left(\Ical_T'^{-1/2}\vect{R}_T'^{-1}\Ical_T'^{-1/2}\otimes \Ical_R^{-1/2}\vect{R}_R^{-1}\Ical_R^{-1/2}\right)\right.\right.\nonumber\\&+&\left.\left.\left(\Ical_T'^{-1/2}\vect{P'}^{H}\vect{S'}_Q^{-1}\vect{P'}\Ical_T'^{-1/2}\otimes \Ical_R^{-1/2}\vect{S}_R^{-1}\Ical_R^{-1/2}\right)\right]^{-1}\right\},\nonumber\\
\end{eqnarray}
where $\vect{R}_T'=\vect{R}_T^T$ with eigenvalue decomposition $\vect{U}_T'\vect{\Lambda}_T'\vect{U}_T'^H$. This objective function subject to the training energy constraint ${\mathrm{tr}(\vect{P}'\vect{P'}^{H})}\leq \mathcal{P}$ seems very difficult to minimize analytically unless special assumptions are made.
\begin{itemize}
  \item $\vect{R}_R=\vect{S_R}$: Then, (\ref{eq:detMMSE-1}) becomes
  \begin{eqnarray}\label{eq:detMMSE-2}
&\mathrm{tr}&\left\{\left(\Ical_T'^{-1/2}\vect{R}_T'^{-1}\Ical_T'^{-1/2}+\Ical_T'^{-1/2}\vect{P'}^{H}\vect{S'}_Q^{-1}\vect{P'}\Ical_T'^{-1/2}\right)^{-1}\right.\nonumber\\ &\otimes& \left.\Ical_R^{1/2}\vect{R}_R\Ical_R^{1/2}\right\}.
\end{eqnarray}
Using once more the fact that $\mathrm{tr}\left(\vect{A}\otimes\vect{B}\right)=\mathrm{tr}\left(\vect{A}\right)\mathrm{tr}\left(\vect{B}\right)$ for square matrices $\vect{A}$ and $\vect{B}$, it is clear from (\ref{eq:detMMSE-2}) that the optimal training matrix can be found by minimizing
\begin{eqnarray}\label{eq:detMMSE-3}
\mathrm{tr}\left\{\left(\vect{R}_T'^{-1}+\vect{P'}^{H}\vect{S'}_Q^{-1}\vect{P'}\right)^{-1}\Ical_T'\right\}.
\end{eqnarray}
Again, here some special assumptions may be of interest.
\begin{itemize}
\item $\Ical_T=\vect{I}$: Then the optimal training matrix can be found by straightforward adjustment of Proposition 2 in \cite{KatselisKT:07}.
\item $\vect{R}_T^{-1}=\Ical_T$: Then (\ref{eq:detMMSE-3}) takes the form
\begin{eqnarray}\label{eq:detMMSE-4}
\mathrm{tr}\left\{\left(\vect{I}+\vect{R}_T'^{1/2}\vect{P'}^{H}\vect{S'}_Q^{-1}\vect{P'}\vect{R}_T'^{1/2}\right)^{-1}\right\}.
\end{eqnarray}

Using the same majorization argument as in the previous Appendix for $\mathrm{tr}\left(\vect{A}^{-1}\right)=\sum_i 1/\lambda_i\left(\vect{A}\right)$, and adopting the notation therein, we should select $\vect{U}=\vect{U}_T'$ and $\vect{V}=\vect{V}_Q$. With these choices, the optimal power allocation problem becomes
\begin{align} \label{eq:detMMSE-5}
\begin{array}{cl}
\minimize{\pi,\varpi,\kappa_{i},i=1,2,\ldots, n_T} & \sum_{i=1}^{n_T}\frac{1}{1+\frac{\pi\left((\vect{\Lambda}_T')_{i,i}\right)\kappa_{i}}{\varpi\left((\vect{\Lambda}_Q)_{i,i}\right)}} \\ %
\text{s.t.}                                 &\sum_{i=1}^{n_T}\kappa_{i}\leq \mathcal{P}  %
\end{array}
\end{align}
where $(\vect{\Lambda}_T')_{i,i}, i=1,2,\ldots,n_T$ are the eigenvalues of $\vect{R}_T'$. Fixing the permutations $\pi(\cdot)$ and $\varpi(\cdot)$, we set
$\gamma_i=\pi\left((\vect{\Lambda}_T')_{i,i}\right)/\varpi\left((\vect{\Lambda}_Q)_{i,i}\right), i=1,2,\ldots,n_T$. With this notation, the problem of selecting the optimal $\kappa_i$'s becomes
\begin{align} \label{eq:detMMSE-6}
\begin{array}{cl}
\minimize{\kappa_{i},i=1,2,\ldots, n_T} & \sum_{i=1}^{n_T}\frac{1}{1+\gamma_i\kappa_{i}} \\ %
\text{s.t.}                                 &\sum_{i=1}^{n_T}\kappa_{i}\leq \mathcal{P}.  %
\end{array}
\end{align}
Following similar steps as in the proof of Proposition 2 in \cite{KatselisKT:07}, we define the following parameter
\begin{eqnarray}\label{eq:mstar}
m_{*}&=&\max\left\{m\in \{1,2,\ldots, n_T\}: \sqrt{\frac{1}{\gamma_k}}\cdot\right. \nonumber\\ &&\left.\sum_{i=1}^{m}\sqrt{\frac{1}{\gamma_i}}-\sum_{i=1}^{m}\frac{1}{\gamma_i}< \mathcal{P}, k=1,2,\ldots,m\right\}.\nonumber\\
\end{eqnarray}
Then, it can be easily seen that for $j=1,2,\ldots, m_{*}$ the optimal $(\vect{D}_{P'})_{j,j}$ is given by the expression
\begin{eqnarray}
\sqrt{\frac{\mathcal{P}+\sum_{i=1}^{m_{*}}\frac{1}{\gamma_i}}{\sum_{i=1}^{m_{*}}\sqrt{\frac{1}{\gamma_i}}}\sqrt{\frac{1}{\gamma_j}}-\frac{1}{\gamma_j}},\nonumber
\end{eqnarray}
while $(\vect{D}_{P'})_{j,j}=0$ for $j=m_{*}+1,\ldots, n_T$.
\end{itemize}

With these expressions for the optimal power allocation, the objective of (\ref{eq:detMMSE-5}) equals
\[
n_T-m_{*}+\frac{\left(\sum_{i=1}^{m_{*}}\frac{1}{\sqrt{\gamma_i}}\right)^2}{\mathcal{P}+\sum_{i=1}^{m_{*}}\frac{1}{\gamma_i}}
\]
and therefore the problem of determining the optimal orderings $\pi(\cdot),\varpi(\cdot)$ becomes
\begin{align} \label{eq:detMMSE-7}
\begin{array}{cl}
\minimize{\pi,\varpi} & n_T-m_{*}+\frac{\left(\sum_{i=1}^{m_{*}}\frac{1}{\sqrt{\gamma_i}}\right)^2}{\mathcal{P}+\sum_{i=1}^{m_{*}}\frac{1}{\gamma_i}}. \\ %
\end{array}
\end{align}

The last problem seems to be difficult to solve analytically. Nevertheless, a simple numerical exhaustive search algorithm, namely Algorithm \ref{alg:optOrd}, can solve this problem\footnote{For easiness, we use the MATLAB notation in this table.}.

\begin{algorithm}
\caption{Optimal ordering for the eigenvalues of $\vect{R}_T'$ and $\vect{S}_Q'$, when $\vect{R}_R=\vect{S}_R$ and $\vect{R}_T^{-1}=\Ical_T$.} %
\label{alg:optOrd} %
\algsetup{indent=1.5em} %
\begin{algorithmic}[1]

\REQUIRE $n_T, B$ such that $B\geq n_T$, $\mathcal{P}$, a row vector $\vect{\lambda}_T'$ containing all $(\vect{\Lambda}_T')_{i,i}$'s for $i=1,2,\ldots,n_T$ in
any order and a row vector $\vect{\lambda}_Q$ containing all $(\vect{\Lambda}_Q)_{i,i}$'s for $i=1,2,\ldots,B$ in
any order.

\STATE \label{alg:step0} Create two matrices $\vect{\Pi}_T$ and $\vect{\Pi}_Q$ containing as rows all possible
permutations of $\vect{\lambda}_T'$ and $\vect{\lambda}_Q$, respectively. Define also the matrix $\vect{\Gamma}=[\ \ ]$.

\LOOP
\STATE for $l=1:n_T!$
\LOOP
\STATE for $t=1:B!$

\STATE $\vect{\Gamma}=\left[\vect{\Gamma}; \vect{\Pi}_T(l,:)./\vect{\Pi}_Q(t,1:n_T)\right]$.
\ENDLOOP
\ENDLOOP

\LOOP
\STATE For each row of $\vect{\Gamma}$ determine the corresponding $m_{*}$ and place it in the corresponding row of a new vector $\vect{M}$.
\ENDLOOP

\LOOP
\STATE for $l=1:n_T!B!$
\STATE \[J(l)=n_T-M(l)+\frac{\left(\sum_{i=1}^{M(l)}\frac{1}{\sqrt{\vect{\Gamma}(l,i)}}\right)^2}{\mathcal{P}+\sum_{i=1}^{M(l)}\frac{1}{\vect{\Gamma}(l,i)}}\]
\ENDLOOP

\STATE $[{\rm val},{\rm ind}]=\min{J}$

\IF{mod$({\rm ind},B!)==0$}

\STATE $j=B!$

\ELSE
\STATE $j=$mod$({\rm ind},B!)$
\ENDIF

\STATE $i=({\rm ind}-j)/B!+1$

\STATE The optimal $\pi(\cdot)$, say $\pi_{\rm opt}$, corresponds to $\vect{\Pi}_T(i,:)$ and the optimal $\varpi(\cdot)$, say $\varpi_{\rm opt}$, to $\vect{\Pi}_Q(j,:)$.
\end{algorithmic}
\end{algorithm}

Note that given the fact that $n_T$ and $B$ are small in practice, the complexity of the above algorithm and its necessary memory are not crucial. However,
as $n_T$ and $B$ increase, complexity and memory become important. In this case, a good solution may be to order the eigenvalues of $\vect{R}_T'$ in decreasing
order and those of $\vect{S}_Q'$ in increasing order. This can be analytically justified based on the fact that for a fixed $m_{*}$, the objective function of
problem (\ref{eq:detMMSE-7}), say ${\rm MSE}(\gamma_1,\ldots,\gamma_{m_{*}})$, has negative partial derivatives with respect to $\gamma_i,i=1,2,\ldots, m_{*}$ and it
is also symmetric, since any permutation of its arguments does not change its value. This essentially shows that a good solution may maintain as active $\gamma$'s the largest possible, through the selection of $m_{*}$. Additionally, the structure of ${\rm MSE}(\gamma_1,\ldots,\gamma_{m_{*}})$ reveals the fact that for every new active $\gamma$, something less than $1$ is added to the MSE, while an inactive value corresponds to adding $1$ to the MSE. This is intuitively appealing with the spatial diversity of MIMO systems and the usual properties that optimal
training matrices possess in such systems (i.e., that they tend to fully exploit the available spatial diversity). The largest possible $\gamma$'s can be achieved with a decreasing order of the eigenvalues of $\vect{R}_T'$ and an increasing order of the eigenvalues of $\vect{S}_Q'$. In this case, it can be  checked that $m_{*}$ can be found as follows
\begin{eqnarray}
m_{*}&=&\max\left\{m\in \{1,2,\ldots, n_T\}: \sqrt{\frac{1}{\gamma_m}}\cdot\right. \nonumber\\ &&\left.\sum_{i=1}^{m}\sqrt{\frac{1}{\gamma_i}}-\sum_{i=1}^{m}\frac{1}{\gamma_i}< \mathcal{P}\right\}.\nonumber
\end{eqnarray}

\item If the modal matrices of $\vect{R}_R$ and $\vect{S}_R$ are the same,
$\Ical_T=\vect{I}$ and $\Ical_R=\vect{I}$, then the optimal training is given by \cite{Bjornson:10}, as these assumptions
correspond to the problem solved therein.

\item In any other case (e.g., if $\vect{R}_R\neq\vect{S_R}$), the (optimal) training can be found using numerical methods like the semidefinite relaxation
approach described in \cite{KatselisRHB:11}. Note that this approach can handle also general $\Ical_{adm}$, not necessarily Kronecker-structured.

\end{itemize}

\bibliographystyle{IEEEtran}
\bibliography{IEEEabrv,pilotrefs,cristian}
\end{document}